\PassOptionsToPackage{greek,english}{babel}
\documentclass[a4paper, thm-restate, cleveref]{lipics-v2021}

\usepackage[utf8]{inputenc}
\usepackage{mathrsfs}
\usepackage{graphicx, svg, hyperref, amssymb, amsmath, amsfonts, xr, breqn, fancyhdr, mathrsfs, amsthm, bbm, mathtools, longtable, ltxtable, todonotes}
\usepackage[english]{babel}
\RequirePackage{hyperref}
\usepackage{thmtools,thm-restate}
\usepackage{stylesheet}
\usepackage{algorithm}
\usepackage{algpseudocode}
\usepackage{hyperref}

\usepackage{tikz}
\usetikzlibrary{calc}
\usepackage{thmtools}
\usepackage{hyperref}
\usetikzlibrary{shapes,positioning,intersections,quotes}
\usepackage[english]{babel} 
\usepackage{blindtext}
\usepackage{etoolbox}
\usepackage{mathtools}

\hypersetup{
	colorlinks=true,
	linkcolor=blue,
	filecolor=violet,  
	citecolor=violet,
	urlcolor=cyan
}%

\title{Approximation Schemes for Geometric Knapsack for Packing Spheres and Fat Objects\thanks{A preliminary version of the work appeared in the proceedings of the 51st EATCS International Colloquium on Automata, Languages, and Programming (ICALP) 2024; \href{https://doi.org/10.4230/LIPIcs.ICALP.2024.8}{doi:10.4230/LIPIcs.ICALP.2024.8}.}} 

\theoremstyle{plain}

\theoremstyle{plain}
\newtheorem{lem}[theorem]{\protect\lemmaname}
\theoremstyle{plain}
\newtheorem{prop}[theorem]{\protect\propositionname}
\theoremstyle{definition}
\newtheorem{defn}[theorem]{\protect\definitionname}
\providecommand{\definitionname}{Definition}
\providecommand{\lemmaname}{Lemma}
\providecommand{\propositionname}{Proposition}
\providecommand{\theoremname}{Theorem}
\global\long\def\C{\mathcal{C}}%
\global\long\def\OPT{\mathsf{OPT}}%
\global\long\def\G{\mathcal{G}}%
\global\long\def\sml{\mathrm{sml}}%
\global\long\def\P{\mathcal{P}}%
\global\long\def\N{\mathbb{N}}%

\global\long\def\DP{\mathrm{DP}}%

\newcommand{\arir}[1]{}
\newcommand{\ari}[1]{#1}
\newcommand{\ag}[1]{#1}

\newcommand{\agtwo}[1]{#1}

\hideLIPIcs

\author{Pritam Acharya}{Department of Mathematics, Indian Institute of Science Education and Research Pune, India}{pritam.acharya@students.iiserpune.ac.in}{}{}

\author{Sujoy Bhore}{Department of Computer Science and Engineering, Indian Institute of Technology Bombay, India}{sujoy@cse.iitb.ac.in}{}{}

\author{Aaryan Gupta}{Department of Computer Science and Engineering, Indian Institute of Technology Bombay, India}{aaryanguptaxyz@gmail.com}{}{}

\author{Arindam Khan}{Department of Computer Science and Automation, Indian Institute of Science Bengaluru, India}{arindamkhan@iisc.ac.in}{}{}

\author{Bratin Mondal}{Department of Computer Science and Engineering, Indian Institute of Technology Kharagpur, India}{mondalbratin2003@kgpian.iitkgp.ac.in}{}{}

\author{Andreas Wiese}{Department of Mathematics, Technical University of Munich, Germany}{andreas.wiese@tum.de}{}{}

\authorrunning{Acharya et al.} 

\Copyright{Pritam Acharya, Sujoy Bhore, Aaryan Gupta, Arindam Khan, Bratin Mondal, Andreas Wiese} 




\category{} 

\funding{Arindam Khan’s research is supported in part by Google India Research Award, SERB Core Research Grant (CRG/2022/001176) on “Optimization under Intractability and Uncertainty”, and the Walmart Center for Tech Excellence at IISc (CSR Grant WMGT-23-0001). A part of this work was done when Aaryan Gupta, Bratin Mondal, and Pritam Acharya were undergraduate interns at Indian Institute of Science, supported by Arindam Khan's Pratiksha Trust Young Investigator Award.}

\acknowledgements{We thank Fl{\'{a}}vio Keidi Miyazawa for providing us with helpful comments and useful references.}

\nolinenumbers 

\EventEditors{}
\EventNoEds{51}
\EventLongTitle{51st EATCS International Colloquium on Automata, Languages, and Programming (ICALP)}
\EventShortTitle{ICALP 2024}
\EventAcronym{ICALP}
\EventYear{2024}
\EventDate{July 8--13, 2024}
\EventLocation{Tallinn, Estonia}
\EventLogo{}
\SeriesVolume{}
\ArticleNo{}

\begin{document}

\maketitle

\begin{abstract}
We study the geometric knapsack problem in which we are given a set of $d$-dimensional objects (each with associated profits) and the goal is to find the maximum profit subset that can be packed non-overlappingly into a given $d$-dimensional (unit hypercube) knapsack.
Even if $d=2$ and all input objects are disks, this problem is known to be \textsf{NP}-hard [Demaine, Fekete, Lang, 2010].
In this paper, we give polynomial time $(1+\eps)$-approximation algorithms for the following types of input objects in any constant dimension $d$:
\begin{itemize}
 \item disks and hyperspheres,
 \item a class of fat convex polygons that generalizes regular $k$-gons for $k\ge 5$ (formally, polygons
 with a constant number of edges, whose lengths are in a bounded range, and in which each angle is strictly larger than $\pi/2$),
 \item arbitrary fat convex objects that are sufficiently small compared to the knapsack.
\end{itemize}
We remark that in our \textsf{PTAS} for disks and hyperspheres, we output the computed set of objects, but for a $O_\eps(1)$ of them, we determine their coordinates only up to an exponentially small error. However, it is unclear whether there always exists a $(1+\eps)$-approximate solution that uses only rational coordinates for the disks' centers. We leave this as an open problem that is related to well-studied geometric questions in the realm of circle packing.

\end{abstract}

\section{Introduction}

\label{sec:intro} One of the cornerstones of geometry is the problem
of packing circles and spheres into a container, e.g., a square or
a hypercube. It dates back to the 17th century when Kepler conjectured
his famous bound on the average density of any packing of spheres
in the three-dimensional Euclidean space~\cite{kepler2010six}. The problem has been investigated, for example, by Lagrange~\cite{chang2010simple}
who solved it in the setting of two dimensions, by Hales and Ferguson~\cite{hales2006formulation}
who proved Kepler's original conjecture, and by Viazovska~\cite{viazovska2017sphere}
who studied the problem in dimension $8$ and was awarded the Fields
medal in 2022 for her work.

A natural corresponding optimization question is the \emph{geometric
knapsack problem}, where we are given a set of $d$-dimensional objects
for some constant $d\in\mathbb{N}$, e.g., circles or \mbox{(hyper-)}spheres,
but possibly also other shapes (like squares, pentagons, hexagons,
etc. for the case of $d=2$) \agtwo{with} each of them having a given profit. The goal is to find the
subset of maximum total profit that can be packed non-overlappingly
into a given square or (hyper-)cube. \agtwo{In this work, we consider the translations of the objects but do not allow rotations.} 

Geometric knapsack is a natural mathematical problem and it is well-motivated
by practical applications in several areas, including radio tower
placement~\cite{szabo2007new}, origami design~\cite{lang1996computational},
cylinder pallet assembly~\cite{castillo2008solving,fraser1994integrated},
tree plantation~\cite{szabo2007new}, cutting industry~\cite{szabo2007new},
bundling tubes or cables~\cite{wang2002improved}, layout of control
panels~\cite{castillo2008solving}, or design of digital modulation
schemes~\cite{peikert1992packing}.

The problem is known to be \textsf{NP}-hard, already for $d=2$ and
if all input objects are \ari{axis-aligned} squares or disks \cite{demaine-Origami-NPHard,adamaszek2016hardness}.
This motivates designing approximation algorithms for it. For hypercubes
in any constant dimension $d$, there is a polynomial time $(1+\eps)$-approximation
algorithm is known for any constant $\eps>0$~\cite{Jansen0LS22}, i.e.,
a polynomial time approximation scheme (\textsf{PTAS}). Thus, this is the best
possible approximation guarantee, unless \textsf{P}=\textsf{NP}.

However, for other classes of (fat) objects, the best-known results
either have approximation ratios that are (far) from their respective lower
bounds or they require resource augmentation, i.e., an increase in the
size of the given knapsack. One intuitive reason for this is that
\ari{axis-aligned} squares and cubes can be stacked nicely without wasting space and
the resulting coordinates are well-behaved, while for more general
shapes this might not be the case, see \Cref{fig:corner}. %

\begin{figure}
	\centering \includegraphics[width=1\textwidth]{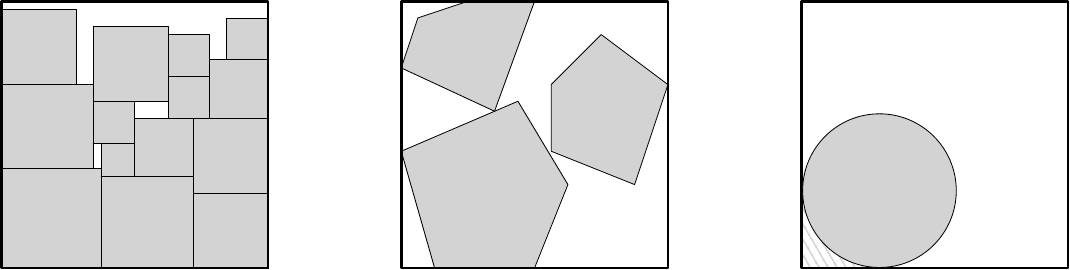}
	\caption{Left: The squares are stacked compactly inside the knapsack. Middle: The pentagons cannot
	be stacked as tightly inside the knapsack as the squares.
	Right: The space in the corner (striped area) cannot be covered by any large circle.}
	\label{fig:corner}
\end{figure}

For circles, the best known result is a $(3+\eps)$-approximation~\cite{Survey-Miyazawa}
but the best-known lower bound is only \textsf{NP}-hardness. 
\ari{Still, the membership in \textsf{NP} is wide open for the following question:
given a set of $n$ circles of $O(1)$ number of different sizes, decide whether they can be packed into a unit square.
 It is also open whether packing circles into a square knapsack is $\exists\mathbb{R}$-complete or not \cite{abrahamsen2020framework}. 
Even for $n$ unit circles,  we do not know the exact value of the smallest size squares that can pack them. 
See \cite{circP}  for the current status of
upper and lower bounds for $n \le 1000$. }

There
is a \textsf{PTAS} in any constant dimension $d$ for this case, but it requires
resource augmentation~\cite{MiyazawaWAOA}. For triangles, there
is a $O(1)$-approximation algorithm (assuming it is allowed to rotate
the triangles arbitrarily) whose precise approximation ratio is not
explicitly specified~\cite{MerinoW20}. Also for this case, it is
still possible that there is a \textsf{PTAS}.
On the other hand, there are settings of geometric knapsack
that do \emph{not} admit a \textsf{PTAS}, e.g., axis-parallel cuboids in three dimensions \cite{ChlebikC09}.

\ari{Furthermore, in practical applications (e.g., loading cargo into a truck or cutting pieces out of raw material like cloth or metal) the objects do not necessarily all have the same shape. For example, Bennell and Oliveira \cite{bennell2009tutorial}  consider a mix of different shapes of objects (their primary objects are circles, rectangles, regular polygons, and convex polygons).
However, the previous papers in the theoretical literature for geometric knapsacks mostly assume that all input objects are of the same type, e.g., only squares, only circles, or only rectangles, etc. Thus, from a theoretical point of view, it is interesting to see how the problem behaves when the input objects might be of different types.}

This raises the following natural question that we study in this paper:
\begin{center}
\emph{What are the best approximation ratios we can achieve for the
geometric knapsack problem, depending on the type of the input objects?
For which type of objects does a \textsf{PTAS} exist?}
\par\end{center}

\subsection{Our contribution}

In this paper, we present a polynomial time $(1+\eps)$-approximation algorithm for the geometric knapsack problem for packing  
$d$-dimensional spheres into $d$-dimensional hypercube knapsack, for any constant dimension $d \ge 2$. For spheres,
there is a complication that possibly any (near-)optimal
\agtwo{packing} for a given instance requires irrational coordinates. Therefore,
our output consists of a set of spheres that can be packed non-overlappingly
inside the given knapsack and whose profit is at least \ag{$(1+\eps)^{-1}w\mathrm{(OPT)}$,
where $w\mathrm{(OPT)}$ denotes the profit of the optimal solution $\mathrm{OPT}$}.
Moreover, for all but at most $O_{\varepsilon}(1)$ spheres, our algorithm
outputs the precise (rational) coordinates of the packing. \footnote{The notation $O_{\eps}(f(n))$ means that the implicit constant hidden in big-$O$ notation can depend on $\eps$.}
For the other $O_{\varepsilon}(1)$
spheres it outputs them up to an exponentially small error in each
dimension. We remark that there are related packing problems for which
it is known that irrational coordinates are sometimes necessary and
that computing them is $\exists\mathbb{R}$-complete (and hence possibly
even harder than \textsf{NP}-hardness) \cite{abrahamsen2020framework}.
On the other hand, if we knew that there always exists a $(1+\eps)$-approximate
solution in which all coordinates are rational with only a polynomial
number of bits, our algorithm would find such coordinates in polynomial
time. We stress that our returned set of spheres is always guaranteed
to fit into the given knapsack with appropriate (possibly irrational)
coordinates but \emph{without} resource augmentation.

Our second result is a polynomial time $(1+\eps)$-approximation
\agtwo{algorithm} for the geometric knapsack problem for wide classes of
convex geometric polygons. We give a \textsf{PTAS} for a class of
fat convex polygons which generalizes pentagons, hexagons, and regular
$k$-gons for constant $k>4$ (see \Cref{fig:regular-obj-packing}). Formally, we require
for each polygon that the angle between any two adjacent edges is
at least $\pi/2+\delta$ for some constant $\delta>0$ and that each
polygon has a constant number of edges with similar lengths (up to
a constant factor). Note that in contrast to many prior results, we
allow each input object to have a different shape, e.g., with a different
number of edges, different angles formed by them, and a different
orientation. Also, the polygons may differ arbitrarily in size.

\begin{figure}
	\centering \includegraphics[width=0.5\textwidth]{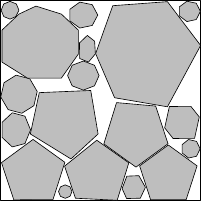}
	\caption{Packing of fat convex polygons in a knapsack}
	\label{fig:regular-obj-packing}
\end{figure}

If each input object is sufficiently small compared to the knapsack,
we obtain
even a polynomial time $(1+\eps)$-approximation for \emph{arbitrary}
fat convex \ari{objects} in any constant dimension $d$. We remark that for other packing problems like
one-dimensional \textsc{Knapsack} or \textsc{Bin Packing}, near-optimal
solutions can easily be achieved via greedy algorithms if the input
objects are sufficiently small. 
\ari{Even for sufficiently small $d$-dimensional axis-aligned hypercuboids, it is known that simple algorithms like NFDH \cite{NFDH, bansal2006bin} have negligible wasted space. 
However, for other geometric objects this
is much harder since we might not be able to place the input objects
compactly without wasting space. 
For example, classical result by Thue \cite{chang2010simple} showed that one can pack at most 
$\frac{\pi}{2 \sqrt{3}} \approx 0.9069$ fraction of the total area, even in the case of packing of unit circles.  
Furthermore, for circles and other similar convex objects, irrational coordinates may arise in the packing and the optimal solution may use a very complicated packing to minimize the wasted space.}

\subsection{Our techniques}
Now, we discuss the techniques of our results, starting with our \textsf{PTAS}
for spheres. To compute our packing, we first enumerate all the large
spheres in the optimal solution, i.e., the spheres whose radius is
at least a constant fraction of the side length of the knapsack. Also,
we guess their placement up to a polynomially small error, which yields
a small range of possible placements for each of them. Note that we
cannot guess these coordinates precisely, since we cannot even exclude
that they are irrational. However, we guarantee that such coordinates
\emph{exist}, by solving a system of polynomial equations \emph{exactly
}in polynomial time.
{Polynomial inequalities (semi-algebraic sets) have previously been applied in the context of circle packing problems \cite{szabo2005global}. 
For instance, Miyazawa et al. \cite{MiyazawaPSSW16} employed polynomial inequalities for the circle bin packing problem. 
We remark that we introduce the new idea of first guessing the coordinates up to a small error, thanks to which we know that we placed the spheres almost at their correct coordinates.}

Next, we want to place small spheres into the remaining part of the
knapsack. Unfortunately, we do not know precisely which part of the
knapsack is available for them since we do not know the precise coordinates
of the large spheres. Thus, there is some area of the knapsack that
is \emph{maybe} used by the large spheres in our packing; however, potentially, the
optimal solution uses it for placing small spheres. Our key insight
is that this area is small compared to the area that is for
sure \emph{not} used by large spheres in the optimal solution. Using the
fact that objects are spheres, we show that some area in each corner
of the knapsack cannot be covered by any large sphere in \emph{any}
solution and whose size is at least a constant fraction of the knapsack
(see the bottom-left empty corner in \Cref{fig:corner}). We use this area to compensate for the fact that we
do not know the precise coordinates of our large spheres and we waste
space because of this.

When we select and pack the small spheres, we define a constant number
of (small) identical knapsacks that fit into the given knapsack together
with the large spheres and into which we place our small spheres.
For the remaining task of placing the small spheres, we argue that it is sufficient to have an algorithm
that uses resource augmentation, i.e., that increases the size of
each knapsack by a factor of $1+\eps$ (in each dimension).
Thus, on a high level, we reduce the problem of packing arbitrary spheres
into \emph{one} knapsack to the problem of packing small spheres into
a \emph{constant number} of knapsacks \emph{with resource augmentation}.

{In order to solve the remaining problem, we use a hierarchical decomposition,
similar to, e.g., \cite{ErlebachJS05, JansenS08, MiyazawaPSSW16, MiyazawaWAOA}. In particular, Miyazawa et al. \cite{MiyazawaPSSW16} employ such a decomposition for the circle bin packing problem and argue that it is possible to first pack the relatively large circles, partition the remaining space in the bins into smaller sub-bins for smaller circles, and continue recursively.
Furthermore, for the circle knapsack problem
Miyazawa et al.~\cite{MiyazawaWAOA} introduced a sophisticated strategy based on a hierarchical decomposition that uses a combination of two configuration-based integer programs (IP). 
In this work, we use a rather bottom-up approach and provide a dynamic program based solution, which is arguably simpler than the corresponding algorithm in \cite{MiyazawaWAOA} and we also show that our subroutine works for arbitrary convex fat objects which is a larger class of objects.
Intuitively, in the structured packing the grid cells are
partitioned such that each placed object $P$ is contained in a constant number of grid cells
whose size is comparable to $P$. Importantly, these grid cells are used \emph{exclusively} by $P$
and not by any other placed object (not even partially). 
Our DP has a subproblem
for each combination of a level (corresponding to a size range of the input objects) and a number of
available grid cells corresponding to this level.
Given such a subproblem, it suffices to enumerate a polynomial number of possibilities for selecting and placing
objects of this level, which reduces the given subproblem
to a subproblem corresponding to the next level. This DP might have applications in other related packing problems.

In our algorithm for fat convex polygons (with the properties described
above), we extend our algorithm for spheres as follows. For the
guessed large polygons, we compute their coordinates \emph{exactly
}in polynomial time. Here, we use the (known) fact that there exists
a placement for them that corresponds to an extreme point solution
of a suitable linear program, which has rational coordinates. Then,
intuitively we use the condition for the polygons' angles to ensure
that the large objects leave a certain area of the knapsack empty.
We use this empty area in a similar way as in the setting of circles.
Again, we place the small objects into a constant number of knapsacks
under resource augmentation, using our new subroutine described above.

If all input objects are sufficiently small compared to the size of the knapsack (formally, we assume that each of them
fits in a smaller knapsack with side length $\Theta(\eps)$)
there are no large objects and, hence,
we can omit the step of enumerating them. In particular, we do not
need the conditions of the polygons' edges anymore.
Since the input objects are so small,
we can show that by losing a factor of $1+\eps$ in the approximation ratio,
we may pretend that we have resource augmentation available.
Hence, we can directly call our subroutine for small objects under resource augmentation.

We leave it as an open question to determine whether irrational coordinates
are sometimes necessary for $(1+\eps)$-approximate 
solutions for geometric knapsack for spheres, i.e., if values $\eps>0$ exist for which this is the case.
If yes, it would be
interesting to determine the best possible approximation ratio one
can achieve with rational coordinates only. Note that this question
is related to the well-studied problem of determining the size of
the smallest knapsack needed to pack a given number of unit circles.
For that problem, it is known that for some number of unit circles
the smallest knapsack has irrational edge lengths~\cite{friedman-cirinsqu}.
On the other hand, recall that if rational coordinates with a polynomial number of bits
always suffice, our algorithm for spheres can compute the coordinates
of \emph{all} returned spheres of our $(1+\eps)$-approximation algorithm
\emph{exactly}.

\subsection{Other related work}
{Two related problems are the geometric bin packing and geometric strip packing problems.
In geometric bin packing, we want to place a given set of geometric objects non-overlappingly into the smallest possible number
of unit-size bins. 
In the strip packing problem, the objective is to pack a set of geometric objects non-overlappingly  into a strip of unit
width and minimum height.
Miyazawa, Pedrosa, Schouery, Sviridenko, and Wakabayashi \cite{MiyazawaPSSW16} gave an \textsf{APTAS} for the circle bin packing problem with resource augmentation, and
an \textsf{APTAS} for the circle strip packing problem. Their framework extends to a wide range of geometric packing problems, e.g., with ellipses, regular polygons, $d$-dimensional spheres under $L_p$-norm, etc.
They also posed the question of whether it is always possible to obtain a rational solution to the problem of packing a set of circles with rational radii in a non-augmented bin of rational coordinates. 
Lintzmayer, Miyazawa, and Xavier~\cite{lintzmayer2018two} used the structure of
solutions presented by \cite{MiyazawaPSSW16}, in combination with new ideas, to obtain a \textsf{PTAS} for the circle knapsack problem
under resource
augmentation in one dimension, assuming that the profit of each circle
equals its area. 
This was recently improved by Chagas, Dell\textquoteright Arriva, and Miyazawa~\cite{MiyazawaWAOA} to the earlier mentioned \textsf{PTAS} under resource augmentation in one dimension for spheres with arbitrary profits in any constant
dimension $d$. 
In addition, there have been many attempts to develop heuristics and
other optimization methods on circle packing, see e.g.,~\cite{szabo2007new,hifi2009literature,lubachevsky1997curved}.}

Recently, independently and parallel to our work, Chagas, Dell\textquoteright Arriva, and Miyazawa~\cite{ChagasNew} have also obtained similar results as the results we present in this paper. They have given a \textsf{PTAS} for the hypersphere multiple knapsack problem. Additionally, they can handle common constraints such as conflict constraints (some pairs of items cannot be packed together), multiple-choice constraints (given a subset of items, at most one of them can be packed), and capacity constraints (items additionally have associated weights, and the sum of weights of the packed items cannot exceed the weight capacity of the knapsack). Moreover, they extended the results to a wide range of convex fat objects, such as, ellipsoids, rhombi, and hyperspheres under $L_p$-norms. 
Furthermore, they provided a resource augmentation scheme for fat objects where rotations are allowed. Finally, they 
show that their framework can be extended to other problems, such as 
the cutting stock problem, the minimum-size bin packing problem, and the multiple strip packing problem.}

Geometric packing problems are also well-studied for rectangular objects. 
For geometric knapsack for axis-parallel rectangles (i.e., when $d=2$),
the best known polynomial time algorithm has an approximation ratio of $17/9+\eps$~\cite{galvez2017approximating}. 
There is a pseudo-polynomial time algorithm with a ratio of $4/3+\eps$
\cite{galvez2021improved} and a pseudo-polynomial time approximation scheme if we require guillotine-separable packing \cite{KhanMSW21}. If it is allowed to rotate the rectangles
by 90 degrees, there is also a polynomial time $(1.5+\eps)$-approximation
algorithm known~\cite{galvez2017approximating}. Moreover, the problem
admits a \textsf{QPTAS} if the input data are quasi-polynomially bounded integers~\cite{adamaszek2014quasi}.
For the settings of geometric bin packing with squares or (hyper-)cubes \cite{bansal2006bin} or skewed rectangles \cite{KhanS23},  
the problem \agtwo{admits} an asymptotic \textsf{PTAS}. In the
case of general rectangles, the best-known result is an asymptotic $1.405$-approximation~\cite{bansal2014improved}
but an asymptotic \textsf{PTAS} cannot exist unless $\mathsf{P=NP}$~\cite{ChlebikC09}.
Maximum independent set in geometric intersection graphs \cite{ChanH12, Mitchell21, GalvezKMMPW22} is another well-studied related problem. 

In a recent paper, Abrahamsen, Miltzow, and Seiferth~\cite{abrahamsen2020framework}
developed a framework to show that for many combinations of allowed
pieces, containers, and motions, the resulting packing problem is
$\exists\mathbb{R}$-complete. For example, they showed that it is
$\exists\mathbb{R}$-complete problem to decide if a set of convex
polygons with at most seven corners each can be packed into a square
if arbitrary rotations are allowed. 
However, it is not known if the setting of packing circles into a
square knapsack is $\exists\mathbb{R}$-complete.

There is also a large body of work on questions about the optimal packings of unit circles into unit squares or equilateral triangles.
We refer to \cite{fekete2023packing, fodor2003densest, goldberg1971packing, hifi2009literature} for an overview.

\subsection{Organization of this paper}
In Section \ref{sec:ptas-ra}, we discuss the \textsf{PTAS} when the input items are sufficiently small fat convex objects. In Section \ref{sec:spheres}, we give our \agtwo{algorithm} for spheres.
In Section \ref{sec:polygon}, we consider the case of convex polygons.  Finally, in Section \ref{sec:conc} we end with conclusions. 
Some proofs have been presented in the appendix. 
The corresponding lemmas and theorems are marked with~$(\star)$.

\section{\textsf{PTAS} under Resource Augmentation}
\label{sec:ptas-ra}
In this section, we present a \textsf{PTAS} when the input items are fat convex objects, and we are allowed to increase the size of the given knapsack by a factor of $1+\eps$ in each dimension. Chagas et al.~\cite{MiyazawaWAOA} presented a \textsf{PTAS} for circles with resource augmentation in one dimension. Their result is based on a combination of multiple integer programs with variables for different configurations for packing parts of the given knapsack.
Our result is arguably simpler and purely based on dynamic programming.

Let $P_i$ be a two-dimensional convex object and let $r_{i}^{\mathrm{out}}(P_i)$ and $r_{i}^{\mathrm{in}}(P_i)$ be
the radius of the smallest circle containing $P_{i}$ and the radius
of the largest circle contained in $P_i$, respectively. 
We will drop $P_i$ when it is clear from the context. 
We say that $P_i$ is $f$\emph{-fat} if  $r_{i}^{\mathrm{out}}/r_{i}^{\mathrm{in}}\le f$
for some value $f \ge 1$. In the remainder of this section, we prove the following theorem.

\begin{theorem}
	\label{thm:fat-ra}
	Let \ari{$f\ge 1$}, $\eps>0$, and $d\in \N$ be constants. Given a set of $d$-dimensional $f$-fat convex input objects, there exists a polynomial time algorithm that can pack a subset of them with a total profit of \ag{$w(\OPT)$} into a knapsack $K':=[0, 1+\eps]^d$, where \ag{$w(\OPT)$} is the optimal profit that can be packed into a knapsack $K:=[0, 1]^d$.
\end{theorem}
\arir{OPT is value or set}

For simplicity, we first describe our algorithm in the setting where $d=2$.
Given a packing of a set of $f$-fat objects $\mathcal{P}=\{P_1, P_2, \dots, P_n\}$ in our knapsack $K=[0, 1] \times [0, 1]$, we want to show that there is also a structured packing of these objects \ari{into an augmented knapsack $K'=[0, 1+\eps]\times [0, 1+\eps]$}, defined via a discrete grid. Let $\dc>0$ be a constant to be defined later such that $1/\dc \in \mathbb{N}$. We place a two-dimensional grid inside \ari{$K'$} such that each grid cell has an edge length of $\dc$. Let $\mathcal{G}$ denote the set of all resulting grid cells. We assume first that each object $P_i \in \mathcal{P}$ is \emph{$\dl$-large}, meaning that $r^{in}(P_i) \ge \dl$ for some given constant $\dl>0$. We say that our given packing for $\mathcal{P}$ is {\em discretized} if there is a partition of $\mathcal{G}$ into sets $\{\mathcal{G}_{P_1}, \mathcal{G}_{P_2}, \dots, \mathcal{G}_{P_n}\}$ such that for each $P_i \in \P$, we have that $P_i$ is contained in the union of the cells in $\mathcal{G}_{P_i}$. Therefore, each object $P_i \in \P$ has its ``own'' set of grid cells $\mathcal{G}_{P_i}$ that contain $P_i$ and that do not intersect with any other object $P_j \in \P \setminus \{P_i\}$. 
{This is similar to the approach used by Miyazawa et al. \cite{MiyazawaPSSW16} for the circle bin packing problem which is also based on a hierarchical decomposition.}

\ari{We show that for an appropriate choice of $\dc$, there is a discretized packing for $\mathcal{P}$ in $K'$, i.e., if we can increase the size of our knapsack $K$ by a factor of $1+\eps$ in each dimension.}

\begin{lemma}
	\label{lem:fatdisc}
	For each $f \ge 1$, $\eps>0$, and $\dl>0$, there is a value $\dc>0$ such that for any set of $\dl$-large $f$-fat convex objects $\mathcal{P}$
	that can be placed non-overlappingly inside a knapsack $K=[0, 1] \times [0, 1]$, there is a discretized packing
	for $\mathcal{P}$ inside knapsack $\ag{K'= } [0, 1+\eps] \times [0, 1+\eps]$
	based on a grid in which each edge of each grid cell has a length of $\dc$.
\end{lemma}

\begin{proof}
	\begin{figure}
		\centering \includegraphics[width=0.6\textwidth]{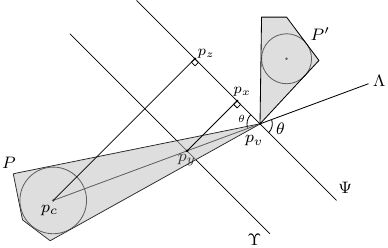}
		\caption{\ag{Line $\Psi$ separates the two $f$-fat and convex objects $P$ and $P'$. We construct line $\Upsilon$ such that it intersects no common grid cells with line $\Psi$. We proceed to shrink the two objects $P, P'$ by a factor of $1+\eps$ such that they cannot intersect the space between lines $\Upsilon$ and $\Psi$.}}
		\label{fig:hyperplane}
	\end{figure}

	\begin{figure}
		\centering \includegraphics[width=0.6\textwidth]{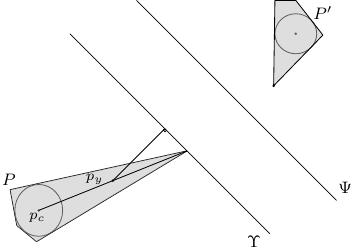}
		\caption{\ag{$P,P'$ are shrunk so that they cannot intersect the space between the lines $\Upsilon$ and $\Psi$.}}
		\label{fig:zoomed}
	\end{figure}
	
	Let $P, P' \in \mathcal{P}$ be two $f$-fat convex $\dl$-large objects packed inside the knapsack.
	Then by {\em separating hyperplane theorem for convex objects}~\cite{boyd2004convex}, we know that there is a line $\Psi$ containing a point $p_v$ on \agtwo{the} boundary of $P$, and $\Psi$ separates $P$ from $P'$  (see Figure \ref{fig:hyperplane}). 
	Now, intuitively, increasing the size
	of the knapsack by a factor of $1+\eps$ is equivalent to shrinking
	the objects in $\mathcal{P}$ by a factor of $1+\eps$. So, we want to find the right constraints such that after shrinking $P$ and $P'$ do not share any grid cell. 
	Let the center of the incircle (of radius $r^{in}$) contained in $P$ be $p_c$ and the line joining $p_c$ and $p_v$ be $\Lambda$. 
	Let the foot of the image of the point $p_c$ on line $\Psi$ be the point $p_z$.  Now the length of line segment $\overline{p_c p_z}  := |\overline{p_c p_z}| \geq r^{in}$ and $|\overline{p_c p_v}| \leq 2 f r^{in}$, as the diameter of $P$ is at most $2 f r^{in}$. Hence, the angle $\theta$ between $\Psi$ and $\Lambda$ is at least ${\sin}^{-1}(\frac{1}{2f})$. Consider points $p_x$ on $\Psi$ and $p_y$ on $\Lambda$ such that $|\overline{p_x p_y}| = \sqrt2 \dc$ and the line  joining $p_x, p_y$ is parallel to the line joining $p_c p_z$. Let $\Upsilon$ be the line passing through $p_y$ and parallel to $\Psi$.  
	Then any point on $\Psi$ does not share a grid cell with any point on $\Upsilon$. 
	Also, $|\overline{p_y p_v}| =  \frac{|\overline{p_c p_v}|}{|\overline{p_c p_z}| } \cdot  {|\overline{p_x p_y}| } \le \frac{2 f r^{in}}{r^{in}} \sqrt2 \dc \le 2 \sqrt2 f \dc$.
	Now we want to shrink  $P$ by $(1+\eps)$ factor keeping $p_c$ at the same position such that the shrunk version of $P$ lies completely within one side of $\Upsilon$ \ag{(see Figure \ref{fig:zoomed})}.
	After shrinking, $\overline{p_c p_v}$ gets smaller by $\eps |\overline{p_c p_v}| \ge \eps r^{in} \ge \eps \dl$.
	Now we choose $\dc$ such that $\dc \le \frac{\eps}{2 \sqrt2 f} \cdot \dl$.
	Thus we satisfy $\eps |\overline{p_c p_v}| \ge \eps \dl \ge  2 \sqrt2 f \dc  \ge  |\overline{p_y p_v}|$ and this ensures that the shrunk down version of $P$ and $P'$ do not share any grid cell.  
	We assign each polygon to the grid cells that it intersects with. 
	Hence this process leads to a discretization such that the no grid cell is intersected by two polygons.
\end{proof}

\begin{figure}
	\begin{center}
		\includegraphics[width=14cm]{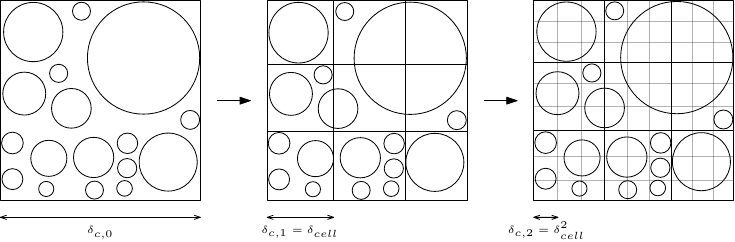}
		\caption{An illustration of a hierarchical grid. The second level gridcells in the example define a `discretized' packing of the circles, i.e., no two circles share a common gridcell.}
		\label{fig:discretized_packing_example}
	\end{center}
\end{figure}

Next, we argue that there is also a structured packing for fat objects
that are not necessarily all (relatively) large. Let $\dl,\dc,\ds>0$
be constants to be defined later (they will depend on $\eps$
which will denote the amount by which we increase the size
of our knapsack). We place now a \emph{hierarchical} two-dimensional grid
with multiple levels. We define that
the whole knapsack $K$ is one grid cell of level 0 of side length
$\delta_{c,0}:=1$. For each level $\ell\ge1$, we define grid cells
whose edges all have a length of $\delta_{c,\ell}:=\dc\delta_{c,{(\ell-1)}}$.
Recursively, for each $\ell\ge1$ 
we partition each grid cell of level $\ell-1$ into ${1/\dc^{2}}$
grid cells of level $\ell$ with side length $\delta_{c,\ell}$ each. See Figure  \ref{fig:discretized_packing_example}.
Similarly as before, we want that there is a partition of the grid
cells such that for each object $P\in\P$ there is a set of grid cells
$\G_{P}$ that contain $P$ and that are disjoint from the grid cells
$\G_{P'}$ for each object $P'\in\P\setminus\{P\}$. Also, we want
that all grid cells in $\G_{P}$ are of the same level, that their
size is comparable to the size of $P$, and that the number of grid
cells in $\G_{P}$ is bounded. To ensure this, we group the objects
$P\in\P$ according to their respective values $r^{\ag{in}}(P)$ which
we use as a proxy for their sizes. Formally, we define $\delta_{small,0}:=1$
and for each level $\ell\ge1$ we define $\delta_{large,\ell}:=\dl{\delta_{c,{(\ell-1)}}}$
and $\delta_{small,\ell}:=\ds\delta_{c,{(\ell-1)}}$. %

For each level $\ell$ we define
\begin{itemize}
 \item $L_{\ell}$ to be all objects
$P\in\P$ with $r^{\ag{in}}(P)\in(\delta_{large,\ell},\delta_{small,{(\ell-1)}}]$;
intuitively, they are ``large'' for level $\ell$,
\item $M_{\ell}$ to be all objects $P\in\P$ with $r^{in}(P)\in(\delta_{small,\ell},\delta_{large,\ell}]$.
\end{itemize}

\begin{figure}
    \centering
    \includegraphics[width=\textwidth]{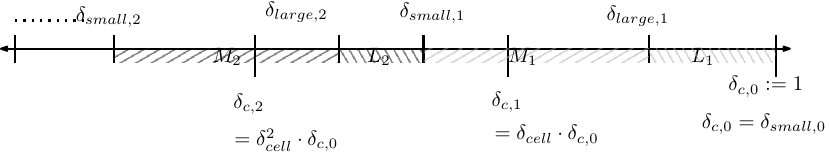}
    \caption{Grouping of objects based on \(\delta_{cell},\delta_{small},\delta_{large}\) and the corresponding Large and Medium sets. The line is not to scale.}
    \label{fig:large-medium-discretized}
\end{figure}

Note that our grid and the sets $L_{\ell}$ and $M_{\ell}$ depend
on the (initial) choice of $\dl, \ds$ and $\dc$. (See Figure~\ref{fig:large-medium-discretized} for an illustration.) In the next lemma, we show via
a shifting argument that there are choices for these values such that the
total area of the objects in $\bigcup_{\ell}M_{\ell}$ is very small.
This will allow us later to pack them separately via a simple greedy
algorithm. In particular, these choices are from a set of $O_\eps(1)$ candidate values $D$, and thus
we will be able to guess them later easily.

\begin{lemma}\label{lem:medium-items}Let $f\ge 1$. There is a global set $D$ with $|D|\le O_{\eps,f}(1)$
such that for any set of $f$-fat convex objects $\P$ that can be packed
in a knapsack $K=[0,1]\times[0,1]$, there are values $\dl,\ds,\dc>0$
that are all contained in $D$
such that for the resulting hierarchical grid and the corresponding
sets $\left\{ L_{\ell},M_{\ell}\right\} _{\ell}$ we have that the
total area of all objects in $\bigcup_{\ell}M_{\ell}$ is bounded
by $\eps$ (and thus corresponding in-radius is less than $\eps$).
\end{lemma}

\begin{proof}
This follows from a shifting argument similar to Lemma \ref{lem:shifting}. 
Say, we need $\ds \le \dc \le \dl$ such that $\ds \le \beta \dc, \dc \le \gamma \dl$, for some constants $\beta, \gamma \in (0,1]$. 

Define $\rho_0 := 1$, 
and $\rhoi{k} := \beta \gamma \cdot \rho_{k-1}$. For each integer $k \in [2/\eps]$ 
\arir{to make sure $\dl, \dc$ is sufficiently small such that medium items have small radius and can fit in the augmented region of $\eps$ width. Another simpler option might be just to show the area of medium to be $\eps^2$ instead of $\eps$, by using more classes.
We can also start with $\rho_0=\eps^2$ but then $\delta_{c,\ell} \neq \dc \delta_{c,\ell-1}$ for all $\ell$.}
we define the following sets of items as shown in \Cref{fig:split}:
\begin{align*}
\mathcal{C}_k = \{C_i \in \mathcal{P} \mid r^{in}(C_i) \in \bigcup_{j=1}^{\infty} (\rho_{k},\rho_{k-1}]\cdot (\beta \gamma)^{(2j/\eps)}\}.
\end{align*}

Then, by averaging argument there is a $k \in \{\frac{1}{\eps}+1, \frac{1}{\eps}+2, \ldots, \frac{2}{\eps}\}$  such that the area of $\mathcal{C}_k$ is at most $\eps$.
Define $\dl:=\rho_{k-1}, \dc:=\gamma \dl, \ds:=\beta \dc=\beta \gamma \dl=\beta \gamma \rho_{k-1}=\rho_k$. 
Define $D:= \{{(\beta \gamma)}^{2/\eps}, {(\beta \gamma)}^{2/\eps -1 }\gamma, {(\beta \gamma)}^{(2/\eps)-1}, \dots,  (\beta \gamma)^{1/\eps+1},  (\beta \gamma)^{1/\eps} \gamma, (\beta \gamma)^{1/\eps})\}$.
Clearly, any chosen $\ds, \dc, \dl$ belong to $D$.
Also, as $k \in \{\frac{1}{\eps}+1, \frac{1}{\eps}+2, \cdots, \frac{2}{\eps}\}$, we have $\dc,\dl \le {(\beta \gamma)}^{1/\eps}$. 

Later we will see that choosing \ag{$\beta=\eps^2/16, \gamma=\eps/(72f)$} suffices for our purpose (see proof of \Cref{lem:fatdiscrete}). 
This also ensures that any object in $\cup_{\ell} M_\ell$ has outradius (and thus inradius)  at most $\eps$. 
\end{proof}
\begin{figure}
	\centering \includegraphics[width=\textwidth]{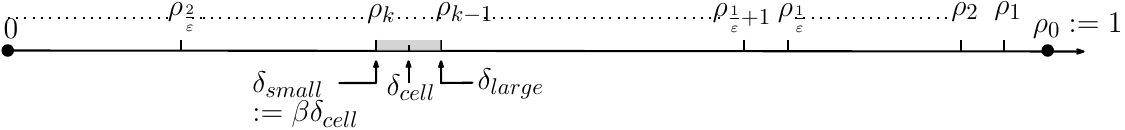}
	\caption{Splitting of sets $\mathcal{C}_k$ on the real line (not to scale) as used in \Cref{lem:medium-items}.}
	\label{fig:split}
\end{figure}

We generalize now our notion of discretized packings. Intuitively,
like before, we require that there is a partition of the grid cells
such that for each object $P\in\P$ there is a set of grid cells $\G_{P}$
that contain $P$ and that are disjoint from the grid cells $\G_{P'}$
for each object $P'\in\P\setminus\{P\}$. Formally, we define that
our packing of $\mathcal{P}$ is \emph{discretized} if
\begin{itemize}
\item for each level $\ell$ and for each object $P\in\mathcal{P}\cap L_{\ell}$
there is a set of $O(1/\dc^{2})$ grid cells
$\mathcal{G}_{P}$
of level $\ell$ such that $P$ is contained in $\mathcal{G}_{P}$,
and there is a single grid cell of level $\ell-1$ that contains all
grid cells in $\mathcal{G}_{P}$, and
\item for any two objects $P,P'\in\mathcal{P}$ (not necessarily of the
same level) and for any two grid cells $C\in\mathcal{G}_{P}$ and
$C'\in\mathcal{G}_{P}'$ their relative interiors are disjoint.
\end{itemize}
We show that by increasing the size of our knapsack by a factor of
$1+\eps$, there is a discretized packing for all objects in {$\bigcup_{\ell}L_{\ell}$}.
As mentioned above, we will pack the objects in {$\bigcup_{\ell}M_{\ell}$}
separately later. %

\begin{lemma}
	\label{lem:fatdiscrete} Let each $f\ge1$ and $\eps>0$.
There is a global set $D$ with $|D| \le O_{\eps,f}(1)$ such that for
any set of $f$-fat objects $\mathcal{P}$ that can be placed non-overlappingly
inside a knapsack $K=[0,1]\times[0,1]$, there is a choice for the
grid with parameters $\dl,\ds,\dc>0$ such that all these values are
contained in $D$ and there is a discretized packing for $\mathcal{P}\cap (\bigcup_{\ell}L_{\ell})$
inside an (augmented) knapsack \ari{$[0,1+O(\eps)]\times[0,1+O(\eps)]$}. \end{lemma}

\begin{proof}
	We start with the given packing of $\mathcal{P}$ in $K$ and do a
	sequence of refinements which leads to our discretized packing for
	$\mathcal{P}\cap\bigcup_{\ell}L_{\ell}$. First, we use the increased
	size of the knapsack to ensure that for each level $\ell$ and any
	two objects $P,P'\in L_{\ell}$, the distance between $P$ and $P'$
	is at least $2\delta_{c,\ell}$. Intuitively, increasing the size
	of the knapsack by a factor of $1+\eps$ is equivalent to shrinking
	the objects in $\mathcal{P}$ by a factor of $1+\eps$. Therefore,
	we can achieve this required minimum distance of $2\delta_{c,\ell}$
	by choosing $\dc$ \ag{appropriately according to the multiple constraints given in the proof of  \Cref{lem:fatdisc}}.

	Next, we would like that for each level $\ell$, each object in $\mathcal{P}\cap L_{\ell}$
	is contained in a grid cell of level $\ell-1$. This might not be
	the case, however, via a shifting argument (giving the grid a random
	shift) we can argue that this is the case for almost all objects in
	$\mathcal{P}$.
 Any object in $\mathcal{P}\cap L_{\ell}$ has diameter at most $2f \delta_{small,\ell-1}$. 
	\ag{Hence, the probability that an object in level $\ell$ intersects a grid line from level $\ell-1$ is at most \ari{$8f\delta_{small,\ell-1}/\delta_{cell,\ell-1}= 8f\ds/\dc$}.
		Let this probability be smaller than $\eps^2/2$, leading to a constraint \ari{$\ds \leq \eps^2\dc/16f$} on the choice of $\dc,\ds$. Then the total area of intersecting objects must also be smaller than $\eps^2$ as the area of all packed objects can be at most the area of \agtwo{the} augmented knapsack.
	}
 
	Thus we can easily pack these intersected objects into extra
	space that we gain via increasing the size of the knapsack (i.e.,
	for a second time).
	
	After this preparation, we process the objects $P\in\mathcal{P}$
	level by level and define their corresponding sets $\mathcal{G}_{P}$,
	starting with the highest level. Consider a level $\ell$. For each
	object $P\in\mathcal{P}$ of level $\ell$ we define $\mathcal{G}_{P}$
	to be the set of all grid cells of level $\ell$ that intersect with
	$P$. Due to our minimum distance between any two objects in $\mathcal{P}$
	of level $\ell$, for any two different objects $P,P'\in\mathcal{P}$
	of level $\ell$ we have that $\mathcal{G}_{P}\cap\mathcal{G}_{P'}=\emptyset$.
	Now it could be that a cell $\texttt{C} \in\mathcal{G}_{P}$ for some $P\in\mathcal{P}$
	intersects not only with $P$, but also with another object $P'\in\mathcal{P}$
	of some deeper level $\ell'>\ell$. We call such a cell $\texttt{C}$ \emph{problematic};
	recall that we wanted the cells in $\mathcal{G}_{P}$ to be used exclusively
	by $P$. Therefore, we move all objects $P'\in\mathcal{P}$ of some
	level $\ell'>\ell$ that intersect a problematic grid \agtwo{cell} in $\mathcal{G}_{P}$.
	We pack them into extra space that we gain due to resource augmentation.
	We do this operation for all levels $\ell$.
	
	In the process above, we move objects that intersect the problematic grid
	cells. We need to argue that the total area of these moved objects
	is small compared to the size of the knapsack and that, therefore,
	we can pack them into additional space that we gain due to our resource
	augmentation. In particular, we need to argue this globally, over
	all levels. The key insight is that if we define a set $\mathcal{G}_{P}$
	for some object $P\in\mathcal{P}$ as above, then each problematic
	cell $\texttt{C}  \in\mathcal{G}_{P}$ must intersect the boundary of $P$ and,
	since $P$ is fat, the number of problematic cells $ \texttt{C} \in\mathcal{G}_{P}$
	is very small compared to the number of cells $\texttt{C}' \in\mathcal{G}_{P}$
	that are contained in $P$ and, thus, for sure \emph{not }problematic.
	
	\arir{Check these calculations}
	\ari{By the classical {\em Barbier's theorem}, we know  the perimeter of a convex set $P$ of level $\ell$ is at most $\pi \cdot \text{diameter}(P) \le 2\pi r^{out} \le 2\pi f r^{in}$.}
	\ag{A curve of length $\delta_{cell, \ell}$ is bound to be contained inside a circle of radius $\delta_{cell, \ell}$. This implies that this curve can intersect at most $9$ grid cells as any circle of radius $\delta_{cell, \ell}$ can be bounded in a $3\times 3$ grid cells of side length $\delta_{cell, \ell}$.}
	\ag{
		Hence, the number of grid cells (of level $\ell$) $N_1$ that the perimeter can intersect is at most $18\pi f r^{in}/\delta_{cell,\ell}$. }
	\ari{
		On the other hand, $P$ completely contains at least grid cells of area $\pi (r^{in}-2\delta_{cell, \ell})^2$, 
		i.e., the number of such gridcells $N_2$ is at least $\frac{\pi (r^{in}-2\delta_{cell, \ell})^2}{(\delta_{cell, \ell})^2}$.
		We need $N_1 \le \eps N_2$.}
	\ag{Equivalently, we want to show, $\delta_{cell, \ell} \le \frac{\eps}{18f} \cdot \frac{(r^{in}-2\delta_{cell, \ell})^2}{r^{in}}$.
		For this we impose the condition that $\dl \ge \frac{72f}{\eps} \dc$.
		Then, $\frac{\eps}{18f} \cdot \frac{(r^{in}-2\delta_{cell, \ell})^2}{r^{in}} \ge \frac{\eps}{18f} \cdot \frac{(r^{in}/2)^2}{r^{in}}\ge  \frac{\eps}{18f} \cdot \frac{\delta_{large, \ell}}{4} \ge \delta_{cell, \ell}$.}

	Using this, we derive a \emph{global }argumentation, stating that
	the total area of \emph{all} problematic grid cells over all objects
	of all levels is at most an $\eps$-fraction of the area of the
	knapsack. Also, if an object $P'$ of some level $\ell'$ intersects
	a problematic grid cell $\texttt{C}$ of some level $\ell<\ell'$, then $P'$
	is very small compared to $\texttt{C}$. Thus, the total area of these objects
	intersecting a problematic grid cell $\texttt{C}$ is essentially the same
	as the area of $\texttt{C}$. Thus, we can pack all these objects into our
	additional space due to resource augmentation.
	
	Finally, we can afford to increase the space of our knapsack such
	that this additional space is even by a constant factor larger than
	the total area of the objects we need to pack into it. Therefore,
	it is easy to find a discretized packing for them in this extra space.
\end{proof}

One can easily extend the above lemma to the case of $O(1)$ number of knapsacks.
\begin{corollary}
\label{corr:fatdiscrete} Let $f\ge1$, $\eps>0$, and $c_m$ be constants.
There is a global set $D$ with $|D|$ to be $ O(1)$ (the constant is independent of $n$ and only depends on $\eps, f, c_m$) such that for
any set of $f$-fat objects $\mathcal{P}$ that can be placed non-overlappingly
inside $c_m$ square knapsacks with side length 1, there is a choice for the
grid with parameters $\dl,\ds,\dc>0$ such that all these values are
contained in $D$ and there is a discretized packing for $\mathcal{P}\cap (\bigcup_{\ell}L_{\ell})$
inside $c_m$ (augmented) knapsacks of side length $1+\eps$ each. 
\end{corollary}

\noindent \textbf{Algorithm.} \ari{Now we describe our algorithm. First, we {\em correctly} guess (i.e., by brute-force enumeration of all possible cases) the values $\dl,\ds,\dc>0$ from set $D$
 due to Lemma~\ref{lem:fatdiscrete}. 
Note that we still do not know $\bigcup_{\ell} L_{\ell}$ or $\bigcup_{\ell} M_{\ell}$, i.e., which objects are there in the optimal packing. 
So, for each level $\ell$ we define
$\tilde{L}_{\ell}$ to be all input objects
$P$ with $r^{\ag{in}}(P)\in(\delta_{large,\ell},\delta_{small,{(\ell-1)}}]$ and  
$\tilde{M}_{\ell}$ to be all input objects $P$ with $r^{in}(P)\in(\delta_{small,\ell},\delta_{large,\ell}]$.}
\arir{Used $\tilde{L}_{\ell}$ instead of ${L}_{\ell}$. As  ${L}_{\ell}$ refers to items in optimal packing. Please check the whole para.}

Then, we compute an optimal
discretized packing via a dynamic program. Intuitively, our DP computes
an optimal subset of $\bigcup_{\ell}\tilde{L}_{\ell}$ for which there is a discretized
packing. 
We introduce a DP-cell 
$\DP[\ell,m]$ for each combination of a level $\ell$ and a value
$m\in\{1,...,n\}$. This cell corresponds to the subproblem of packing
a maximum profit subset of the objects in $\tilde{L}_{\ell},\tilde{L}_{\ell+1},\tilde{L}_{\ell+2},\dots.$
via a discretized packing into at most $m$ grid cells of level $\ell-1$,
i.e., with side length $\delta_{c,(\ell-1)}$ each. Recall that each
object in $\tilde{L}_{\ell}$ is relatively large compared to the grid cells
of level $\ell-1$. Therefore, we can pack only constantly many items
from $\tilde{L}_{\ell}$ into each of these $m$ grid cells of level $\ell-1$.
Hence, there are only constantly many options for how the set $\mathcal{G}_{P}$
of an object $P\in \tilde{L}_{\ell}$ in the optimal solution to our subproblem
can look like. We say that a \emph{configuration} is a
partition of a grid cell of level $\ell-1$ into sets of grid cells of level $\ell$.
Each grid cell of level $\ell-1$ contains
{$O(1/\dc^2)$ many grid cells of level $\ell$. Hence, there are only constantly
many configurations.
We assume two configurations to be \emph{equivalent} if they are identical up to translation by an integral multiple of $\delta_{c,(\ell-1)}$, i.e., by an integral multiple of the edge length of a grid cell of level $\ell-1$.
Denote by $C$ the total number of resulting equivalence classes.
We guess in time $m^{O(C)}\le n^{O(C)}$ how many grid cells
have each of the at most $C$ configurations (up to equivalences). Then, we assign the
items in $\tilde{L}_{\ell}$ into the grid cells according to this guess.
\ari{We can do this by weighted bipartite matching. For each object
$P\in \tilde{L}_{\ell}$, each possible configuration $\mathcal{G}'$, and each set in the partition of
$\mathcal{G}'$, we
can check easily whether $P$ fits into $\mathcal{G}'$. 
In the bipartite graph, one side will contain the objects in $\tilde{L}_{\ell}$ and the other side will contain the sets in the partition of
$\mathcal{G}'$.
If $P$ fits into set $Q$ in the \agtwo{partition} of $\mathcal{G}'$, then there is an edge with edge cost as $\profit(P)$. }
Our guess
yields a certain number $m'$ of empty grid cells of level $\ell+1$
into which we need to pack items in $\tilde{L}_{\ell+1},\tilde{L}_{\ell+2},\dots$.
We assign these items according to the solution in the DP-cell $\DP[\ell+1,\min\{m',n\}]$.
Note that there are at most $n$ items and, hence, we never need more
than $n$ grid cells of level $\ell+1$. Also, since our input data is
polynomially bounded, the number of classes $\tilde{L}_{i}$ is bounded
by~$n^{O(1)}$. Thus, our DP runs in time $n^{O(C)}$.

Additionally, we pack medium objects from the set $\bigcup_{\ell}\tilde{M}_{\ell}$
separately in a strip of the form $[0,1]\times[1,1+O(\eps)]$. We select the most profitable subset of $\bigcup_{\ell}\tilde{M}_{\ell}$
(up to a factor of $1+\eps$) whose total area is bounded by $\eps$
(see Lemma~\ref{lem:medium-items}).
We replace each
of these objects with the smallest square that contains it, which
increases its area only by a constant factor.
We can pack these squares efficiently into a (slightly larger) strip $[0,1]\times[1,1+O(\eps)]$ using the $\mathtt{NFDH}$ algorithm~\cite{NFDH}. This yields the following lemma.

%

\begin{lemma}
\label{lem:med}
In polynomial time we can compute a set $\P'\subseteq \bigcup_{\ell}\tilde{M}_{\ell}$
and a non-overlapping placement of $\P'$ inside $[0,1]\times[1,1+O(\eps)]$
such that $w(\P')$ is at least the profit of any subset of \ari{$\bigcup_{\ell}\tilde{M}_{\ell}$}
whose total area is at most $\eps$.
\end{lemma}

\begin{proof}
Let the most profitable subset of $\bigcup_{\ell} \tilde{M}_{\ell}$
(up to a factor of $1+\eps$) whose total area is bounded by $\eps$ be $\OPT_{M}$.
 We arrange all items in $\bigcup_{\ell}\tilde{M}_{\ell}$ in order of decreasing profit density and select (prefix) subset of items with highest profit density such that their area is bounded by $2\eps$. Let this subset of items be $M'$. We argue that $w(M')\ge w(\OPT_M)$ as any items in $\bigcup_{\ell} \tilde{M}_{\ell}$ has area at most $\eps$.
 Therefore, our selected set $M'$ has profit at least $\profit(\bigcup_{\ell}M_{\ell})$
(see Lemma~\ref{lem:medium-items}).
Now we pack  $M'$ 
separately in a strip of the form $[0,1]\times[1,1+O(\eps)]$.
We replace each
of these objects with the smallest square that contains it, which
increases its area only by a constant factor of at most {$2f^2$ since each item is $f$-fat. Note that this square has a side length of at most $2r^{out} \le 2f r^{in} \le 2f\eps$ as all the medium objects have an inradius of at most $\eps$. Now we state a result that gives an approximation algorithm for packing small squares of size $O(\eps$) efficiently.}
\begin{lemma}[\cite{NFDH}]\label{lem:nfdh}
	In a rectangle of size $a \times b$, the $\mathtt{NFDH}$ algorithm guarantees to pack an area of at least $(ab-\mu(a+b))$ with a sufficient number of squares of maximum side length $\mu$.
\end{lemma}

We pack these square objects greedily with the $\mathtt{NFDH}$ algorithm
\cite{NFDH} stated in \Cref{lem:nfdh} into a (slightly larger) strip $[0,1]\times[1,1+O_f(\eps)]$ to arrive at the result we wanted to prove.
\end{proof}
The algorithm is easy to extend to the case of a constant number (say, $c_m$) of knapsacks. For a single knapsack, we call DP[0,1], whereas for multiple knapsacks we need to just call DP[0,$c_m$]. This leads to the following corollary. 

\begin{corollary}
	\label{cor:fat-ra}
	Let \ari{$f\ge 1$}, $\eps>0$, $d\in \N$, and $c_m$ be constants. Given a set of $d$-dimensional $f$-fat convex input objects, there exists a polynomial time algorithm that can pack a subset of them with a total profit of \ag{$w(\OPT)$} into $c_m$ knapsacks with side length $1+\eps$, where \ag{$w(\OPT)$} is the optimal profit that can be packed into $c_m$ knapsacks with side length 1.
\end{corollary}

One can easily extend our algorithm above to any constant dimension $d$.
This completes the proof of \Cref{thm:fat-ra}. 

A consequence is that we obtain a polynomial time $(1+\eps)$-approximation
\emph{without }resource augmentation if all input objects are small,
i.e., if $r^{out}(P)\le\eps$ for each given object $P\in\P$.
Using this property, we can argue that there is a $(1+\eps)$-approximate solution
in which only the area $[0, 1-\Theta(\eps)]\times[0, 1-\Theta(\eps)]$ of the
knapsack is used. Thus, we can use the free space for the resource augmentation
that is required by our algorithm due to \Cref{thm:fat-ra}.

\begin{theorem}
\label{lem:small-poly}
Let $d\in \N$ be a constant.
There is a polynomial time $(1+O(\eps))$-approximation
for the geometric knapsack problem if the set of input objects $\P$
consists of convex fat $d$-dimensional objects such that $r^{out}(P)\le\eps$
for each $P\in\P$.
\end{theorem}

\begin{proof}
We begin by assuming that the optimal solution for this problem is $\OPT$. We then remove two randomized strips of width $\eps$.
First, we select an $x,y \in [0,1-\eps]$ uniformly at random. Then we remove all items that intersect the horizontal strip $[0,1]\times [x, \eps]$ and all items that intersect the vertical strip $[x, \eps] \times [0,1]$. 
Each item is then cut with probability at most $2(2 r^{out}+\eps)$. 
In expectation, then we incur a loss of at most $2(2 r^{out}+\eps) \le 6 \eps$ fraction of the profit. 
This can also be derandomized by an averaging argument. 
From this, we get a packing of $(1-O(\eps)) \ag{w(\OPT)}$ profit in a $(1-\eps)\times (1-\eps)$ bin. We now use Theorem~\ref{thm:fat-ra} to get this packing from the set $\P$ in polynomial time. 
\end{proof}

\section{Spheres}
\label{sec:spheres}
In this section we present our $(1+\eps)$-approximation algorithm
for the case of $d$-dimensional spheres.
Let $\C = \{ C_1, C_2, \dots , C_n \}$ be a set of  $n$ number of $d$-dimensional hyperspheres. We denote the radius and profit of each hypersphere \(C_i \in \C\) by \(r_i\) and \(\profit_i\). 
For an object $C_i$ we denote its volume (or area in 2-dimension) to be $\area(C_i)$. 
For a collection of objects $\mathcal{A}$, we define its volume and profit to be $\area(\mathcal{A}):=\sum_{C_i \in A} \area(C_i)$ and $\profit(\mathcal{A}):=\sum_{C_i \in A} \profit(C_i)$, respectively.
We are given a unit knapsack $K := [0,1]^d$.

We first consider the case of circles, i.e., $d=2$.
Let $\OPT$ be an optimal solution and $\C_{\OPT}$ be the circles in $\OPT$.
Let $\eps \in (0,1/2]$ be a constant and assume that $1/\eps\in\mathbb{N}$. First,
we want to classify the input circles into \emph{small} and \emph{large
}circles such that each large circle is much larger than any small circle.
Due to the following lemma, we can do this such that we can ignore
all circles that are neither large nor small by losing only a factor
of $\eps$. 
We will use this standard shifting argument throughout the paper (see Appendix \ref{sec:shifting} for the details of shifting argumentation). 

\begin{lem}
\label{lem:create-gap} Given a constant $\eps\in (0,1)$ and $1/\eps \in \mathbb{N}$, there is a
	set of global constants $\eps^{(0)},...,\eps^{(1/\eps)}\geqslant 0$ such that
	$\eps^{(j)}=(\eps^{(j-1)})^{24}$ for each $j\in\{1,...,{1/\eps}\}$
	and a value $k\in\{1,...,{1/\eps}-1\}$ with the following property: if we define
	$\epsl:=\eps^{(k)}$ and $\epss:=\eps^{(k+1)}$, 
	then sum of profits of all circles $C_i$ in $\OPT$ with radii $\epss \leq r_i \leq \epsl$ is at most $\eps \cdot w(\OPT)$.
\end{lem}

\begin{proof}
Define $\rho_0=\eps$, and $\rhoi{k} := (\rho_{k-1})^{24}$ for $k \in [1/\eps]$. For each integer $k \in [1/\eps]$,  we define the following sets of items:
$\mathcal{C}_k = \{C_i \in \mathcal{C} \mid r_{i} \in (\rho_{k},\rho_{k-1}]\}$.
Then by pigeonhole principle, one of the sets $\mathcal{C}_{\tau} \cap \OPT$ will contain items of weight at most $\eps w(\OPT)$.
Define $\epsl=\rho_{\tau-1}$ and $\epss=\rho_{\tau}$. The set of $\rho_i$'s defines the set of global constants with cardinality $1/\eps$. 
\end{proof}

We guess the value $k\in\{0,...,1/\eps-1\}$ due to Lemma~\ref{lem:create-gap}.
We define that a circle $C_{i}\in\C$ is \emph{large }if $r_{i}>\epsl$
and \emph{small} if $r_{i}\leqslant\epss$. Also, note that $\epss = \epsl^{24}$. 

\subsection{Guessing large circles} \label{subsec:Large_Guess}

We observe that in $\OPT$ there can be only a constant number of
large circles since each large circle covers a constant fraction of the
available area in the knapsack. 
\begin{prop}
\label{big_sol} Any feasible solution
	can contain at most $(1/\eps)^{O_k(1)}$ large circles. 
	\end{prop}

\begin{proof}
	Each large circle has area $\pi \epsl^2$. Thus there can be at most $1/(\pi \epsl^2)= 1/(\pi ({\eps^{{24}^k})^2})$ large circles. 
\end{proof}

We guess a feasible solution of the large circles in $\OPT$ that satisfy the packing constraints in time $n^{(1/\eps)^{O_k(1)}}$,
denote them by $\C_{L}^{*}$. In related problems, like the two-dimensional
knapsack problem with squares or rectangles, one can easily guess
the correct placement of the guessed large such objects  
(assuming rational input data).
For packing circles, it is not clear that if there is a packing in which the
centers of the circles are placed at rational coordinates.
{However, we can still model the circle packing problem as a semi-algebraic set defined via polynomial inequalities, similar to Miyazawa et al.~\cite{MiyazawaPSSW16}.}

We first guess for each circle $C_{i}\in\C_{L}^{*}$
an \emph{estimate }for its placement in $\OPT$. Denote by $\hat{x}_{i}^{(1)},\hat{x}_{i}^{(2)}\in[0,1]$
the coordinates of the center of $C_{i}$ in $\OPT$. We guess values
$\tilde{x}_{i}^{(1)},\tilde{x}_{i}^{(2)}\in\{0,\frac{\eps}{n},\frac{2\eps}{n},...,1\}$
such that $\hat{x}_{i}^{(1)}\in[\tilde{x}_{i}^{(1)},\tilde{x}_{i}^{(1)}+\frac{\eps}{n})$
and $\hat{x}_{i}^{(2)}\in[\tilde{x}_{i}^{(2)},\tilde{x}_{i}^{(2)}+\frac{\eps}{n})$.
Note that there are only $O(n^{2}/\eps^{2})$ possibilities for
each $C_{i}\in\C_{L}^{*}$, and hence only $n^{(1/\eps)^{O_k(1)}}$
possibilities overall for all circles $C_{i}\in\C_{L}^{*}$.

Given these guessed values $\tilde{x}_{i}^{(1)},\tilde{x}_{i}^{(2)}$
for each circle $C_{i}\in\C_{L}^{*}$, we verify that our guess was
correct or not, i.e., confirm that there exists, indeed a corresponding placement
for each circle $C_{i}\in\C_{L}^{*}$ such that the circles in $\C_{L}^{*}$
do not overlap. Therefore, we define a system of quadratic inequalities
 that \agtwo{describes} the problem of finding such a placement.
We require that this placement is consistent with
our guesses $\tilde{x}_{i}^{(1)},\tilde{x}_{i}^{(2)}$ for each $C_{i}\in\C_{L}^{*}$.
\begin{alignat}{2}
\max\{\tilde{x}_{i}^{(1)},r_{i}\}\leqslant x_{i}^{(1)} & \leqslant\min\left\{ \tilde{x}_{i}^{(1)}+\frac{\eps}{n},1-r_{i}\right\}  & \,\,\,\,\, & \forall C_{i}\in\C_{L}^{*}\nonumber \\
\max\{\tilde{x}_{i}^{(2)},r_{i}\}\leqslant x_{i}^{(2)} & \leqslant\min\left\{ \tilde{x}_{i}^{(2)}+\frac{\eps}{n},1-r_{i}\right\}  &  & \forall C_{i}\in\C_{L}^{*}\label{eq:equations}\\
(x_{i}^{(1)}-x_{j}^{(1)})^{2}+(x_{i}^{(2)}-x_{j}^{(2)})^{2} & \geqslant(r_{i}+r_{j})^{2} &  & \forall C_{i},C_{j}\in\C_{L}^{*}\nonumber \\
x_{i}^{(1)},x_{i}^{(2)} & \geqslant0 &  & \forall C_{i}\in\C_{L}^{*}\nonumber
\end{alignat}

Let $|\C_{L}^{*}|=: t$. Then, the above system has $2t$ variables and $k:=O(t^2)$ constraints. 
It is not clear how to compute a solution to this system in polynomial
time. It is not even clear whether it has a solution in which each variable
has a rational value. However, in polynomial time, we can \emph{decide}
whether it has a solution (without computing the solution itself) using an algorithm from \cite{GrigorevVV88}.

Note that the set of solutions satisfying system (\ref{eq:equations}) is a semi-algebraic set in the field of real numbers. 
Thus, whether a given set of circles can be packed or not (the decision problem) reduces to a decision problem of whether
this semi-algebraic set is nonempty or not. 
Here, each constraint $i \in [k]$ in (\ref{eq:equations}) can be written as a function $f_i (x_{1}^{(1)},x_{1}^{(2)}, \dots, x_{t}^{(1)},x_{t}^{(2)})\geqslant 0$ where each $f_i \in \mathbb{Q}[x_{1}^{(1)},x_{1}^{(2)}, \dots, x_{t}^{(1)},x_{t}^{(2)}]$ is a polynomial with rational coefficients of degree at most two.
Therefore, deciding the circle packing problem is equivalent to deciding whether the following formula is satisfiable or not: $F:= (\exists x_{1}^{(1)})(\exists x_{1}^{(2)}) \dots (\exists x_{t}^{(1)}) (\exists x_{t}^{(2)}) \land_{i=0}^{k} f_i(x_{1}^{(1)},x_{1}^{(2)}, \dots, x_{t}^{(1)},x_{t}^{(2)}) \geqslant 0$. 
To solve this decision problem, we use the following result.

\begin{theorem}[\cite{GrigorevVV88}]
	\label{thm:GrigorevVV88}
	Let $f_1, f_2, \dots, f_k \in \mathbb{Q}[x_{1}^{(1)},x_{1}^{(2)}, \dots, x_{t}^{(1)},x_{t}^{(2)}]$ be polynomials with absolute value of any coefficient to be represented by $M$ bits and maximum degree $\Delta$.
	There is an algorithm that decides whether the formula $F:= (\exists x_{1}^{(1)})(\exists x_{1}^{(2)}) \dots (\exists x_{t}^{(1)}) (\exists x_{t}^{(2)}) \land_{i=0}^{k} f_i(x_1, y_1, \dots, x_n, y_n) \geqslant 0$ is true, with a running time of $M^{O(1)}(k\Delta)^{O(t^2)}$.
	
	If it is true, the algorithm also returns polynomials $f, g_1, h_1, \dots, g_t, h_t \in \mathbb{Q}[x]$ with coefficients of bit size at most $M^{O(1)} (k \Delta)^{O(t)}$ and maximum degree $k^{O(t)}$, such that for a root $x$ of $f(x)$, the assignment $x_{1}^{(1)}=g_1(x), x_{1}^{(2)}=h_1(x), \dots, x_{t}^{(1)}=g_n(x), x_{t}^{(2)}=h_n(x)$
	satisfies the
	formula $F$.
	
	Moreover, for any rational $\alpha>0$, it returns values $\bar{x}_{1}^{(1)},\bar{x}_{1}^{(2)} \dots, \bar{x}_{t}^{(1)},\bar{x}_{t}^{(2)} \in \mathbb{Q}$ such that $| \bar{x}_{i}^{(1)}- {x}_{i}^{(1)} | \leqslant \alpha$ and  $| \bar{x}_{i}^{(2)}- {x}_{i}^{(2)} | \leqslant \alpha$, for $1 \leqslant i \leqslant t$, in time at most $(\log (1/\alpha) M)^{O(1)} (k \Delta)^{O(t^2)}$.
\end{theorem}

We crucially use here that our system has only constantly many variables
and constraints,~i.e., in (\ref{eq:equations}), we have that $\Delta, k, t$ are constants and that $M$ is polynomially bounded in $n$.
From \Cref{thm:GrigorevVV88}, we see that
in polynomial time we can decide whether
(\ref{eq:equations}) has a solution.

If the system (\ref{eq:equations}) does
not have a solution, then we reject this guessed combination of $\C_{L}^{*}$
and values $\tilde{x}_{i}^{(1)},\tilde{x}_{i}^{(2)}$ for each circle
$C_{i}\in\C_{L}^{*}$. We assume in the following that it has a solution.
Observe that the guessed values $\tilde{x}_{i}^{(1)},\tilde{x}_{i}^{(2)}$
yield an estimate for $\hat{x}_{i}^{(1)},\hat{x}_{i}^{(2)}$ up to
a (polynomially small) error of $\eps/n$.

Now we will improve this even to an \emph{exponentially} small error of $1/2^{n/\eps}$.

\subsubsection{Improving the precision of the large spheres}
\label{subsec:Improve_Large} Recall that for each large circle $C_{i}\in\C_{L}^{*}$
we guessed its center in the optimal packing up to a polynomial error
of $\frac{\eps}{n}$. We improve this to only an exponential error
of at most $\frac{1}{2^{n/\eps}}$. To do this, we apply Theorem~\ref{thm:GrigorevVV88}
with $\alpha:=\Theta(\frac{1}{2^{n/\eps}})$. This yields more precise
estimates $\bar{x}_{i}^{(1)},\bar{x}_{i}^{(2)}$ for each $C_{i}\in\C_{L}^{*}$.
There is an important subtlety though: for our guessed coordinates
$\tilde{x}_{i}^{(1)},\tilde{x}_{i}^{(2)}$ we can assume that they
differ from the coordinates of $\OPT$ by at most our polynomial error
of $\frac{\eps}{n}$. For the new estimates $\bar{x}_{i}^{(1)},\bar{x}_{i}^{(2)}$
we can \emph{not }guarantee this: our subroutine from Theorem~\ref{thm:GrigorevVV88}
possibly returns a solution that is (close to) feasible for the
large circles, but not (close to) a solution that is feasible for
the large and for the small circles. Because of this, we guessed the estimates
$\tilde{x}_{i}^{(1)},\tilde{x}_{i}^{(2)}$ for each $C_{i}\in\C_{L}^{*}$,
so that we can assume that these estimates really correspond to $\OPT$
and not just to some arbitrary solution to \eqref{eq:equations}.

If it were true that there is always a $(1+\eps)$-approximate solution
in which the center of each circle has rational coordinates that can
be encoded with a polynomially bounded number of bits, then we could
choose $\alpha$ appropriately to compute it. More precisely, we could
compute a range for each coordinate that contains only one rational
number with a bounded number of bits, and we could compute this number
afterward.

\subsection{Placing small circles}\label{subsec:placing_small}

We want to select small circles from $\C$ and place them inside the
knapsack, so that they do not overlap with each other or
with the circles in $\C_{L}^{*}$. To this end, we define $\epsc := \epsl^{12}$ (i.e., $\epss=\epsl^{12} \epsc = \epsc^2)$ to subdivide the knapsack
into a grid with $1/\epsc^2$ square grid cells of side length $\epsc$ . Our choice of parameters ensures that each small circle is small compared to each grid
cell and each large circle is big compared to each grid cell. Formally, for each $\ell,\ell'\in\{0,1, \dots ,\frac{1}{\epsc}-1\}$
we define a grid cell $G_{\ell,\ell'}:=[\ell\cdot\epsc,(\ell+1)\cdot\epsc)\times[\ell'\cdot\epsc,(\ell'+1)\cdot\epsc)$.
We define the set of all grid cells by $\G:=\{G_{\ell,\ell'}:\ell,\ell'\in\{0,1, \dots ,1/\epsc-1\}\}$.

We say that a placement of a circle $C_{i}\in\C_{L}^{*}$ is \emph{legal} if its
center is placed at a point $(x_{i}^{(1)},x_{i}^{(2)})$ such that
$\max\{\tilde{x}_{i}^{(s)},r_{i}\}\leqslant x_{i}^{(s)}\leqslant\min\left\{ \tilde{x}_{i}^{(s)}+\frac{\eps}{n},1-r_{i}\right\} $
for each $s\in\{1,2\}$. We show that there is a structured packing with near-optimal profit in which each small circle is contained in a grid cell
that does not intersect with any large circle in $\C_{L}^{*}$ in any legal packing of them. This will allow us to decouple the remaining problem
for the small circles from the large circles, even though we do not know the exact placement for the latter.
Moreover, in each grid cell the small circles use only a reduced area of size $(1-\eps) \epsc \times(1-\eps) \epsc$.
Let $\C_{S}^{*}$ denote the small circles in $\OPT$.

\begin{lemma}\label{lem:structured-packing}
	In polynomial time, we can compute a set of grid cells $\G_w$ such that no grid cell in $\G_w$ intersects with any
	circle $C_i \in\C_{L}^{*}$ for any legal placement of $C_i$. Moreover, there is a set of small circles $\C_S \subseteq \C_{S}^{*}$ such that $w(\C_{S}) \geqslant (1-\eps)w(\C_{S}^{*})$ and
	the circles in $\C_S$ can be packed non-overlappingly inside $|\G_w|$ grid cells of size $(1-\eps) \epsc \times(1-\eps) \epsc$ each.
\end{lemma}
We will prove Lemma~\ref{lem:structured-packing} later in Section~\ref{sec:structured-packing}. Using it, we compute an approximation to the set $\C_S$ via \ari{\Cref{cor:fat-ra}.}

We pack the computed circles into our grid cells $\G_w$, and denote them by $\C'_S$. In particular, they do not intersect any of the
large circles in $\C_{L}^{*}$ in any legal placement of them. Our solution (corresponding to the considered guesses) consists of $\C_{L}^{*}\cup \C'_S$. Recall that we can guarantee that these circles can be packed non-overlappingly inside the knapsack. Also, for the circles in
$\C'_S$ we computed their placement \emph{exactly} and for the circles in $\C_{L}^{*}$ we computed their placement up to our polynomially small error of $\eps/n$. As mentioned in Section~\ref{subsec:Improve_Large}, we can then reduce this error to an exponentially small error.

\subsection{Structural packing for small circles}\label{sec:structured-packing}
In this section, we prove Lemma~\ref{lem:structured-packing}.
First, we show that intuitively almost \agtwo{every} small circle in $\C_{S}^{*}$ is contained inside some grid cell.
Formally, we show that the total area of all other small circles in $\C_{S}^{*}$ is small.
For any set $\mathcal{S}$ of circles or grid cells, we define $a(\mathcal{S})$ to be the total area of the elements in $\mathcal{S}$.

\begin{lem}
	\label{lem:small-circles-grid-cells}
	Let $\mathcal{C}_{cut}\subseteq \C_{S}^{*}$ be the set of all small circles in $\C_{S}^{*}$ that intersect more than one grid cell. We have that $\area(\mathcal{C}_{cut}) \leqslant 8 \epss/\epsc \leqslant \eps \epsl^2/64$.
\end{lem}

\begin{proof}
Each small circle has diameter at most $2\epss$. 
	For each horizontal (resp. vertical) gridline $\ell$, each $C \in \mathcal{C}_{cut}$ must lie in a strip of height (resp. width) $4\epss$ and width (resp. height) 1. Totally there are $2/\epsc$ gridlines (see Figure~\ref{fig:cut_cell_illustration}).
    
    \begin{figure}[h!]
	\begin{center}
		\includegraphics[width=9cm]{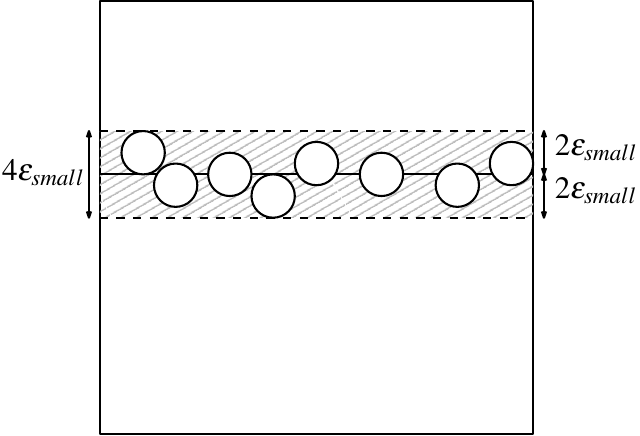}
		\caption{Obtaining the strip of height \(4\eps_{small}.\)}
		\label{fig:cut_cell_illustration}
	\end{center}
    \end{figure}
 
	Thus $\area(\mathcal{C}_{cut}) \leqslant 8 \epss/\epsc = 8 \epsl^{12}  \leqslant  \eps \epsl^2/64$. The last inequality follows from the fact that $\epsl \leqslant \eps \leqslant 1/2$.
\end{proof}

We will repack the circles in $\mathcal{C}_{cut}$ later such that each of them is contained inside one single grid cell.
Thus, for each small circle $C_{i}\in\C_{\OPT} \setminus \C_{cut}$ there is a grid cell $G_{\ell,\ell'}$ for
some $\ell,\ell' \in\{0,1,...,1/\epsc-1\}$ such that $C_{i}$ is contained in $G_{\ell,\ell'}$ in $\OPT$.
When we select and place small circles, we must be careful
that they do not intersect any large circles from $\C_{L}^{*}$. One
difficulty for this is that we do not know the precise coordinates
of the large circles. Therefore, we place small circles only into grid
cells that do not overlap with any large circle from $\C_{L}^{*}$ in
any legal placement of them.
Formally, we partition the cells
in $\G$ into three types: white, gray\label{key}, and black cells (see \Cref{fig:grid_cell_partition}).
\begin{figure}[h]
		\begin{center}
		\includegraphics[width=8cm]{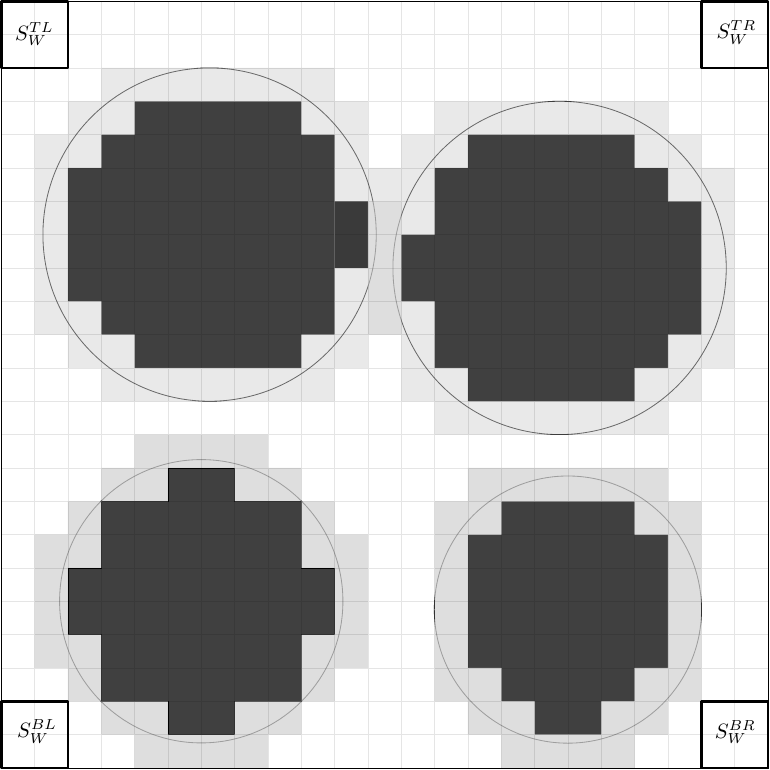}
		\caption{Partitioning grid cells into white, black, and gray cells. Later corner regions $S_W^{BL}, S_W^{TL}, S_W^{BR}, S_W^{TR}$ are used to to pack items in gray cells. }
		\label{fig:grid_cell_partition}
	\end{center}
\end{figure}
\begin{defn} Let $G_{\ell,\ell'}\in\G$ for some $\ell,\ell'\in\{0,1,...,1/\epsc-1\}$.
	The cell $G_{\ell,\ell'}$ is
	\begin{itemize}
		\item \emph{white} if $G_{\ell,\ell'}$ does not intersect with any circle
		$C_{i}\in\C_{L}^{*}$ for any legal placement of $C_{i}$,
		\item \emph{black} if $G_{\ell,\ell'}$ is contained in some circle $C_{i}\in\C_{L}^{*}$
		for any legal placement of $C_{i}$,
		\item \emph{gray }if $G_{\ell,\ell'}$ is neither white nor black.
	\end{itemize}

\end{defn} The gray cells are problematic for us since a gray cell
might be (partially) covered by a large circle in $\C_{L}^{*}$ but
we do not know by how much (and which part of the cell). Therefore,
we do not place any small circles into gray cells. However, $\OPT$
might place small circles into these cells (and obtain the profit of
these circles). On the other hand, we can show that the number of gray
cells is very small, only a small fraction of all grid
cells can be gray. In order to do this, we use \agtwo{the fact} that the values $\tilde{x}_{i}^{(1)},\tilde{x}_{i}^{(2)}$
for each circle $C_{i}\in\C_{L}^{*}$ estimate the placement of each
large circle relatively accurately, and that the grid cells are relatively
small. This allows us to prove that almost all
cells are black or white. Also, we can compute all gray
cells efficiently. Let $\G_{g}\subseteq\G$ denote
the set of all gray grid cells in $\G$.

\begin{lem}
	\label{lem:gray-circles}
	The total area of gray cells $\area(\G_{g})$ is at most $\eps \epsl^2/5$. We can compute $\G_g$ in polynomial time.
\end{lem}

\begin{figure}[ht]
	\begin{center}
		\includegraphics[width=8cm]{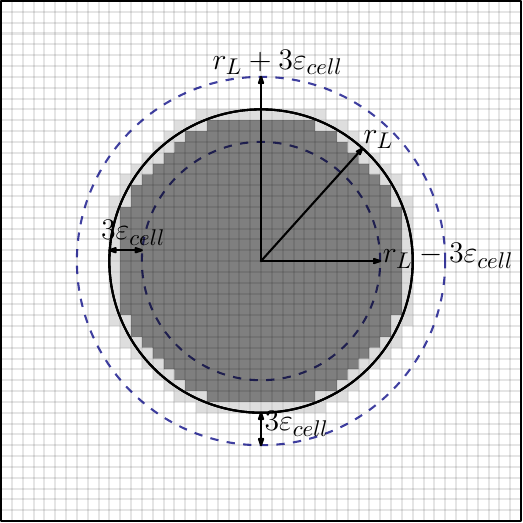}
		\caption{Three types of cells for a particular circle}
		\label{fig:gray_cell_partition_types}
	\end{center}
\end{figure}
\begin{proof}
	Consider a large circle $C_L$ with radius $r_L \geqslant \epsl$. Let $\G_{g, C_L}, \G_{b, C_L}$ be the set of gray and black grid cells that overlap with $C_L$.
	We will first show that $\G_{g, C_L}$ is much smaller compared to  $\G_{b, C_L}$.
	Note that for $C_L$,  we only have an estimate $\tilde{x}_{C}^{(1)},\tilde{x}_{C}^{(2)}$ of the coordinate of its center up to
	a \emph{polynomially} small error of $\eps/n$. 
	
	First, we claim that all the cells intersected by a circle centered at  $(\tilde{x}_{C}^{(1)},\tilde{x}_{C}^{(2)})$ of radius $r_L- 3\epsc$ are in $\G_{b, C_L}$. 
	For contradiction, assume there is a gray cell $S_g$ intersected by a circle $C$ centered at  $(\tilde{x}_{C}^{(1)},\tilde{x}_{C}^{(2)})$ of radius $r_L- 3\epsc$.
	Then, the distance of any point in $S_g$ from the center of the circle is at most $r_L- 3\epsc +\sqrt{2}\epsc < r_L-\eps/n$ for large enough $n$. Thus  $S_g$  is completely contained within $C$.
	
	Similarly, we can show that the cells that are outside of a circle centered at  $(\tilde{x}_{C}^{(1)},\tilde{x}_{C}^{2})$ of radius $r_L+ 3\epsc$ are not in $\G_{b, C_L}$.
	
	Hence, $\area(\G_{b, C_L}) \geqslant \pi (r_L- 3\epsc)^2$ and $\area(\G_{g, C_L}) \leqslant \pi (r_L+3\epsc)^2 - \pi(r_L-3\epsc)^2=12 \pi \epsc r_L$.
	Hence, 
	\[
	\frac{\area(\G_{g, C_L})}{\area(\G_{b, C_L})} \leqslant \frac{12 \pi \epsc r_L}{\pi (r_L- 3\epsc)^2} \leqslant \frac{12 \epsc r_L}{(r_L/2)^2} \leqslant 48\frac{\epsc}{\epsl} = 48 \epsl^{11} \leqslant \eps \epsl^2/5
	\]
	The second inequality follows from the fact that $\epsl \leqslant \eps \leqslant 1/2$ and $r_L \geqslant \epsl$.  The last inequality uses the fact that $\eps \leqslant 1/2$.
	Therefore, $\sum_{C \in \C_{L}^{*}} {\area(\G_{g, C_L})} \leqslant (\eps \epsl^2/5) \cdot \sum_{C \in \C_{L}^{*}}  {\area(\G_{b, C_L})} \leqslant \eps \epsl^2/5$.
\end{proof}

See Figure~\ref{fig:gray_cell_partition_types} for an illustration. 

Unfortunately, it is not sufficient for us that there are
only a few gray cells. It might be that almost all cells are either
gray or black and, hence, we need to place most of the selected small
circles into gray cells (in order to obtain an $(1+\eps)$-approximate
solution).

However, we can show that this is not the case. We prove that the
number of white cells (which we can safely use for small circles) is
at least by a factor $1/\eps$ larger than the number of gray
cells. To show this, we exploit the geometry of the circles. In each
corner of the knapsack, there are cells that cannot intersect with
any large circle, simply because the grid cells are small compared to
the large circles and because of the shape of the large circles (see the corner regions in \Cref{fig:grid_cell_partition}).
Hence, these grid cells are white. Let $\G_{w}\subseteq\G$ denote
the set of all white grid cells in $\G$. 

\begin{lem}
	\label{lem:white-circles}
	The total area of white grid cells $\area(\G_{w})$ is at least $\epsl^2/4$. We can compute $\G_w$ in polynomial time.
\end{lem}
\begin{figure}[ht]
	\begin{center}
		\includegraphics[width=8cm]{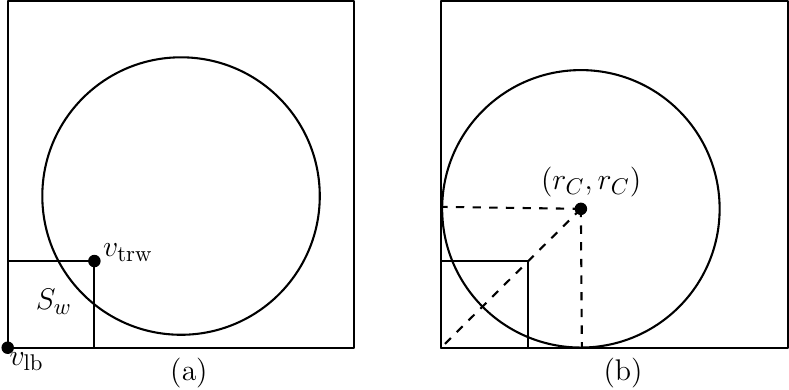}
		\caption{(a) The square \(S_w\) and the large circle \(C\), (b) Orientation of circle \(C\) after translation}
		\label{fig:corner_cell}
	\end{center}
\end{figure}

\begin{proof}
	In \Cref{fig:corner_cell}, consider a corner (w.l.o.g., left bottom corner) $v_{lb}$ of the knapsack. Let  $S_w$ be a square region of side length $\epsl/2$ such that $S_w$ is contained in the knapsack and shares $v_{lb}$.
	We will show that no large circles intersect $S_w$ and thus all grid cells in $S_w$ are white.
	
	For contradiction, assume that some large circle $C$ intersects $S_w$. As $r_C \geqslant \epsl$, if $C$ intersects either the top or right edge of $S_w$ it must contain the top right corner of $S_w$ at $(\epsl/2, \epsl/2)$.  Thus, $(\epsl/2, \epsl/2)=: v_{trw} \in interior(C)$. 
	Then consider $C$ in the knapsack and we can move $C$ towards $v_{lb}$  such that it touches two boundary edges of the knapsack, i.e.,  the center of $C$ is now at $(r_C, r_C)$. Note that, after this translation, $C$ should continue to intersect $S_w$, i.e.,  $(\epsl/2, \epsl/2) \in interior(C)$. 
	However, distance of $v_{trw}$ from $(r_C, r_C)$ is $\sqrt{(r_C-\frac{\epsl}{2})^2 + (r_C-\frac{\epsl}{2})^2} \geqslant \sqrt{2(\epsl-\frac{\epsl}{2})^2} \ge \epsl$.
	Thus,  $v_{trw} \notin interior(C)$. This is a contradiction. 
	
	We argue this analogously for all the corners of the knapsack to complete the proof. 
\end{proof}

Using Lemmas \ref{lem:small-circles-grid-cells}, \ref{lem:gray-circles},
and \ref{lem:white-circles}, we show that there is a $(1+\eps)$-approximate
solution in which each small circle is contained in some white cell;
in particular, no small circle is placed inside a gray cell. To prove
this, we delete all small circles in the $O(\eps |\G_w|)$ white grid
cells with the smallest total profit among all white cells and place all circles from gray cells and
all circles from $\C_{cut}$ into those.

\begin{lem}
	\label{lem:small-circles-white-grid-cells-1}
	There is a set $\C'_S \subseteq \C^*_S$ of small circles with $p(\C'_S)\geqslant (1-O(\eps)) p(\C^*_S)$ such that there is a packing for $\C'_S$ using the grid cells in $\G_w$ only.
	\end{lem}

\begin{proof}
	We delete all circles in the $2\eps$ fraction of white cells with the smallest total
	weight among all white cells, incurring an $2\eps$ fraction loss in the weight.
	This makes $\eps \epsl^2/2$ space available for packing circles in the gray cells and $\C_{cut}$.
	Now, the total volume of circles in the gray cells and $\C_{cut}$ is at most  $\eps \epsl^2/5 + \eps \epsl^2/64 \leqslant \eps \epsl^2/4$. As these circles are small, the available space is enough to pack them even using the \texttt{NFDH} algorithm. 

	Recall that the radius of each
	small circle is at most $\epss$, which is at least by a factor
	$1/\eps$ smaller than the size of each grid cell (which is $\epsc \times\epsc$).
	Therefore, with a shifting argument, we can show that from each grid
	cell $G_{\ell,\ell'}\in\G$ we can delete small circles from $\C^*_S$ contained in
	$G_{\ell,\ell'}$ whose total profit is only an $O(\eps)$-fraction
	of the overall profit of those circles, such
	that the remaining circles in $G_{\ell,\ell'}$ fit even inside a smaller
	grid cell of size $(1-\eps)\epsc \times(1-\eps)\epsc$.
	If we apply this operation to each grid cell in $\G_{w}$, the remaining
	circles fit non-overlappingly into $|\G_{w}|$ (small) knapsacks of
	size $(1-\eps)\epsc \times(1-\eps)\epsc$
	each.
\end{proof}

We complete the proof of Lemma~\ref{lem:structured-packing} by applying the following lemma to $\C_{S}:= \C''_{S}$
which shows that we can sacrifice a factor of $1+O(\eps)$ \agtwo{to be} able to use resource augmentation when we pack the small circles.

\begin{lem}
	\label{lem:grid-resource-augmentation-1}
	There is a set of small circles $\C''_{S}\subseteq\C'_{S}$ such that $w(\C''_{S})\ge(1-\eps)w(\C'_{S})$ and it is possible to place the circles in $\C''_{\sml}$ non-overlappingly inside $|\G_{w}|$ square knapsacks of size $(1-\eps)\epsc \times(1-\eps)\epsc$ each.
	\end{lem}

\begin{proof}
	Consider a square knapsack in $\G_{w}$ of side length $\epsc$.
        Consider a random strip $[x, \min\{\epsc, x+\eps\epsc\}]\times [0,\epsc] \cup [0, \max\{0, x+\eps\epsc -\epsc\}] \times [0,\epsc]$ of width $\eps\epsc$ where $x$ is chosen uniformly at random in $[0,\epsc]$. Then the probability of a small circle of radius $r$ intersected by the strip is at most $(2r+\eps\epsc)/\epsc$.
	By removing all circles intersected by a random strip of height (resp. width) $\epsc$ and width (resp. height) $\eps\cdot\epsc$ in both directions, we incur a loss of at most $2(2 r/\epsc + \eps)$ fraction of optimal profit. As $r\le \epss$, we incur a loss of at most $O(\eps)$ fraction of profit as $\epss = \epsc^2 = \epsl^{24}$.
\end{proof}

\subsection{Higher dimensions}

Our techniques from the previous section extend directly to the problem of packing
hyperspheres in any (constant) dimension $d$ which yields our main theorem for the setting of packing
(hyper-)spheres. See Appendix for the detailed proof. 

\begin{theorem}\label{thm:d-dim-ptas} Let $d\in\N$ be a fixed constant. For the geometric knapsack problem with
	$d$-dimensional hyperspheres, there is a polynomial time algorithm
	that computes a set of hyperspheres $\tilde{\C}$ with $p(\tilde{\C})\geqslant(1-\eps)\OPT$
	that can be placed non-overlappingly inside the knapsack. For all
	but $O_{\eps}(1)$ hyperspheres in $\tilde{\C}$ we compute the
	precise coordinate of the corresponding packing; for the other $O_{\eps}(1)$
	circles we compute an estimate of the packing with an additive error
	of at most $\frac{1}{2^{n/\eps}}$ in each dimension. \end{theorem}

\section{Polygons}
\label{sec:polygon}

{
In this section, we adapt our techniques from the previous sections to obtain a \textsf{PTAS} for the case that each input object is a fat and convex
polygon with a constant number of sides (edges) whose lengths differ by at most a constant factor. Moreover,  
each angle between any two adjacent edges is at least $\pi/2$. Note that, these classes of objects generalize regular polygons with greater than 4 sides. 
}

{Formally, we assume that we are given a set
$\P = \{ P_1, P_2, \dots , P_n \}$ of $n$ polygons that are \emph{$(f,\alpha,q, t)$-well-behaved}, i.e., for each
polygon~$P_{i}\in\P$ we assume that}
\begin{itemize}
	\item $P_i$ is fat, i.e., $r_{i}^{\mathrm{out}}/r_{i}^{\mathrm{in}}\le f$
	for some (global) constant $f \ge 1$, where
	$r_{i}^{\mathrm{out}}$ and $r_{i}^{\mathrm{in}}$ is
	the radius of the smallest circle containing $P_{i}$ and the radius
	of the largest circle contained in $P$, respectively,
	\item the angle between any two consecutive edges of $P_i$ is at least $\pi/2+\alpha$,
	for some (global) constant $\alpha>0$,
	\item $P_i$ has at most $q$ {edges} for some (global) constant $q$
 \item $P_i$ is $t$-regular, i.e., the lengths of any two of its edges
	differ at most by a factor of $t$.
\end{itemize}

For instance, regular pentagons are $(2,\pi/10,5,1)$-\emph{well-behaved}.
For each polygon \(P_i \in \P\) we denote by \(\profit_i\) its profit, and for a set of polygons
$\mathcal{P}' \subseteq \mathcal{P}$ we denote by $\profit(\mathcal{P}'):=\sum_{P_i \in \mathcal{P'}} w_i$ their total profit.
For any {object} $C$ we define its area to be $\area(C)$, and
for any collection of objects $\mathcal{A}$ we denote their {total} area by $\area(\mathcal{A}):=\sum_{P_i \in A} \area(P_i)$.
{We want to pack a subset of $\P$ non-overlappingly into the
unit knapsack $K := [0,1]^2$. We do not allow rotations in our packing.
}

Let $\eps>0$. We require $\eps < g(f,\alpha,q,t)$ for a function $g$, which will be defined later.
In fact, we will choose $\eps < g(f,\alpha,q,t) := \min \{\frac{1}{8f}, \frac{\pi^2 sin^2(\alpha)}{q^2t^2(2+80f)}\}$.
 In contrast to
the case with hyperspheres, we show that we can compute each coordinate of our
packing exactly. We classify each polygon $P_{i}\in\P$ as large or small according
to the respective value $r_{i}^{\mathrm{in}}$.
For this, we define values $\epsl$ and $\epss$.

\begin{lem}\label{lem:create-gap-polygons}
Given a constant $\eps\in(0,1)$, there is a
	set of global constants $\eps^{(0)},\ldots,\eps^{(1/\eps)}\geqslant 0$ such that
	$\eps^{0}=\eps$, and $\eps^{(j)}=(\eps^{(j-1)})^{20}$ for each $j\in\{1,\ldots,{1/\eps}-1\}$
	and a value $k\in\{0,\ldots,{1/\eps}-1\}$ with the following properties. If we define
	$\epsl:=\eps^{(k)}$ and $\epss:=\eps^{(k+1)}$, then by
	losing a factor of
 $1+\eps$ in our approximation
	guarantee, we can assume that each polygon $P_{i}\in\P$ satisfies
	that $r_{i}^{\mathrm{in}} \le \epss$ or $r_{i}^{\mathrm{in}} > \epsl$.

	\end{lem}
\begin{proof}
Define $\rho_0=\eps$, and $\rhoi{k} := (\rho_{k-1})^{20}$ for $k \in [1/\eps]$. For each integer $k \in [1/\eps]$,  we define the following sets of items:
$\mathcal{C}_k = \{C_i \in \mathcal{C} \mid r_{i}^{\mathrm{in}}  \in (\rho_{k},\rho_{k-1}]\}$.
Then by pigeonhole principle, one of the sets $\mathcal{C}_{\tau} \cap \OPT$ will contain items of weight at most $\eps w(\OPT)$.
Define $\epsl=\rho_{\tau-1}$ and $\epss=\rho_{\tau}$. The set of $\rho_i$'s defines the set of global constants. Note that this also ensures $\epsl\le \eps$.
\end{proof}

We guess the value $k\in\{0,\ldots,1/\eps-1\}$ due to Lemma~\ref{lem:create-gap-polygons}
and define that a polygon $P\in\P$ is \emph{large }if $r_{i}^{\mathrm{in}}\ge \eps^{(k-1)}=\epsl$
and \emph{small }if $r_{i}^{\mathrm{in}}<\eps^{(k)}=\epss$. We discard
all input polygons that are neither large nor small.
Similar to the case of hyperspheres, we guess the large polygons in
$\OPT$. Since they are fat, there can be only a constant number of
them.

\begin{prop} 
	\label{prop: finite_circle}
	Any feasible solution can contain at most $(1/\eps)^{O(1)}$
	large polygons. \end{prop}

We now calculate the placement of the large polygons using a linear program. Define $\P_{L}^*$ to be the set of large polygons in $\OPT$.
Note that, since they are convex polygons instead of circles or hyperspheres, we can compute an \emph{exact }placement of the polygons with \emph{rational }coordinates.
For this, we use the following approach, which was also noted by Abrahamsen et al. in~\cite{abrahamsen2020framework}. Consider
the placement of the polygons $\P_{L}^{*}$ in $\OPT$. Each
side $e$ of each polygon $P_{i}\in\P_{L}^{*}$ is contained in a
line $\{x :a_{e}x=b_{e}\}$ for some vector $a_{e}$
and a scalar $b_{e}$. For each corner vertex $v$ of each polygon $P_{j}\in\P_{L}^{*}$,
we have that $a_{e}x\ge b_{e}$ or $a_{e}x\le b_{e}$ (or both); we
guess which of these cases applies. Let $v_{i}$ be a special vertex for each polygon $P_i$ defined as a vertex with the least value of $x_{i}^{(1)}$. Then, the coordinates
$({x}_{i}^{(1)},{x}_{i}^{(2)})$ of the special vertex $v_{i}$
of each polygon $P_{i}\in\P_{L}^{*}$ satisfy a system
of linear inequalities defined as follows.
There are three types of inequalities:
\begin{itemize}
	\item \emph{Positivity constraints}: $x_{i}^{(1)},x_{i}^{(2)} \ge0, \forall P_{i}\in\P_{L}^{*}\nonumber$
	
	\item \emph{Packing constraints}: $\forall P_{i},P_{j} \in \P_{L}^{*}$ any vertex $v$ in polygon $P_{j} \neq P_{i}$ cannot lie inside $P_{i}$. A vertex $v$ lies inside polygon $P_{i}$ if it satisfies the inequalities described above.
	
	\item \emph{Container constraints}:  $\forall P_{i}\in\P_{L}^{*}\nonumber$, $0 \leq x_{i}^{(1)} \leq a_i$ and $b_i \leq x_{i}^{(2)} \leq 1-c_i$, where $a_i, b_i$, and $c_i$ can be calculated exactly in constant time for a given polygon. They represent the maximum value of $x_{i,v}^{(1)}-x_{i}^{(1)}$, maximum value of $x_{i}^{(2)}-x_{i,v}^{(2)}$, and maximum value of $x_{i,v}^{(2)}-x_{i}^{(2)}$, where $x_{i,v}^{(1)},x_{i,v}^{(2)}$ vary over all vertices $v$ of polygon $P_{i}$.
	
\end{itemize}

From Proposition \ref{prop: finite_circle}, we know that there can only be at most $(1/\eps)^{O(1)}$
large polygons in $\OPT$. We take all possible subsets $\P_L ^*$ of this size and smaller from the set $\P$, which are polynomially many.
We compute a feasible solution of packing these subsets $P_L ^*$ to it which is easy since it has only
$O_{\eps, k, f, \alpha, t, q}(1)$ variables and constraints for each subset for polynomially many subsets, by using the ellipsoid method.

Let $\P_{S}$ be the set of all small polygons. 
Now for the correctly guessed large subset $P_L ^*$, we compute a near-optimal packing of the small polygons $P_S^* \subseteq \P_{S}$. Our goal is to pack the small polygons in the bin with only a loss of $\eps$-fraction of profit, corresponding to the guessed $P_L^*$. We then return the solution $P_L^* \cup P_S^*$ which has maximum weight over all guessed values of $P_L^*$ initially and claim that this packing is near-optimal.

In order to pack small polygons, we need an analogous version of Lemma \ref{lem:structured-packing}.
We define grid cells again similarly such that $\epsc$ is much smaller compared to $\epsl$ and $\epss$ is much smaller compared to $\epsc$.
Let us choose $\epsc= \epsl^{5}$ and $\epss=\epsc^{4}$. Recall that $\epss = \epsl^{20}$.
{Intuitively, since our input polygons are well-behaved, we can prove that a certain amount of space is not used by the large polygons, similar to Lemma~\ref{lem:white-circles}.} To ensure this, we require that $\eps$ is sufficiently small, which in particular also yields a bound on $\epsc$. Using this, we show that there are many grid cells that are disjoint from any large polygon (similarly as the white grid cells in Section~\ref{sec:spheres}).
W.l.o.g.~we assume now on that $\P_{L}^{*}$ be the set of large polygons in $\OPT$ and $\P_{S}^{*}$ be the set of all small polygons in $\OPT$.

\begin{lemma}\label{lem:structured-packing-polygon}
	There is a function $g:\mathbb{R}^4_{\geqslant 0} \rightarrow \mathbb{R}_{\geqslant 0}$ such that if all given polygons are $(f,\alpha,q, t)$-well-behaved and $\eps < g(f,\alpha,q,t)$ then
	\begin{itemize}
	 \item in polynomial time we can compute a set of grid cells $\G_w$ such that no grid cell in $\G_w$ intersects with any
	polygon $P_i \in\P_{L}^{*}$,
	\item there is a set of small polygons $P_S \subseteq \P_{S}^*$ such that $w(P_{S}) \geqslant (1-O(\eps))w(\P_{S}^{*})$ for the optimal packing of large polygons, and
	\item the polygons in $P_S$ can be packed non-overlappingly inside $|\G_w|$ grid cells that have size
	\mbox{$(1-\eps) \epsc \times (1-\eps) \epsc$} each.
	\end{itemize}
\end{lemma}
We will prove Lemma~\ref{lem:structured-packing-polygon} {with similar techniques as we used in the proof of Lemma~\ref{lem:structured-packing}} in Section \ref{susec:smallpoly}.
In order to select and place the small polygons, we use the algorithm due to \Cref{cor:fat-ra}. 

{Let $P_S^*$ denote our computed solution for the small polygons.
We return the solution $\tilde P = P_L^* \cup P_S^*$ which has maximum weight over all initially guessed combinations for the polygons in $P_L^*$ and their approximate coordinates.}
This proves the following theorem. 

\begin{theorem} \label{thm:polygon} For any constants $f,q\ge1$ and $t,\alpha>0$ there is a
	\textsf{PTAS} for the geometric knapsack problem for \emph{$(f,\alpha,q, t)$-well-behaved} polygons.\end{theorem}

\subsection{Proof of \Cref{lem:structured-packing-polygon}}
\label{susec:smallpoly}
In this subsection we prove \Cref{lem:structured-packing-polygon}. First, we remove all small polygons in $\P_{S}^{*}$ that intersect more than one grid cell. The following lemma implies that they fit into very few grid cells.
\begin{lem}
	\label{lem: cut-polygon}
	Let $\mathcal{P}_{cut}\subseteq \P_{S}^{*}$ be the set of all small polygons in $\P_{S}^{*}$ that intersect more than one grid cell. We have that $\area(\mathcal{P}_{cut}) \le 8 f \epss/\epsc$. 
 Also, these items can be packed separately into $2/\epsc$ cells.
\end{lem}

\begin{proof}
	For each horizontal (resp. vertical) gridline $\ell$, each $P \in \mathcal{P}_{cut}$ must lie in a strip of height (resp. width) $4 r^{out} \le 4f r^{in} \le 4f \epss$ and width (resp. height) 1. 

\begin{figure}
	\begin{center}
		\includegraphics[width=14cm]{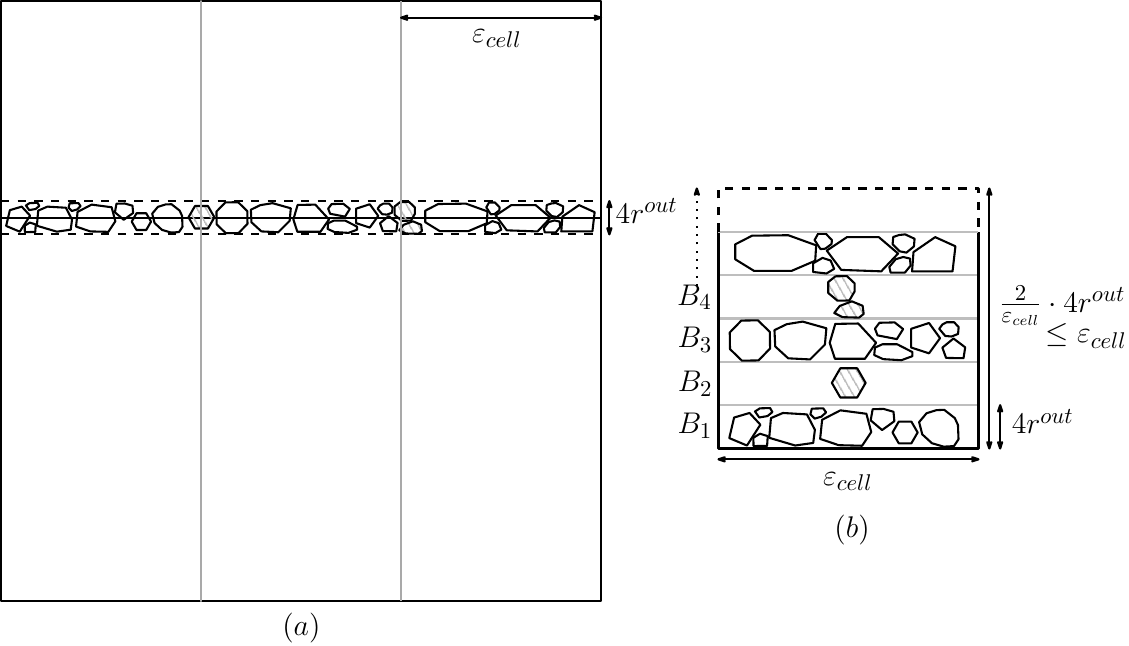}
		\caption{The items in \(\mathcal{P}_{cut}\) that are cut by a grid line can be packed into $2/\epsc$ boxes of height $4r^{out}\le 4f\epss$ and width \(\varepsilon_{cell}.\) All such boxes can be stacked and packed into a grid cell of side length $\epsc$.}
		\label{fig:gray_cell_partition}
	\end{center}
\end{figure}
 
All such items in a particular strip can be packed into $2/\epsc$ boxes each with width (resp. height) $\epsc$ and height (resp. width) $4f \epss$. These boxes can be stacked and packed into one grid cell, requiring $4f \epss \cdot \frac{2}{\epsc} \le \epsc$. 
This follows from our assumption $\epss=\epsc^{4}$ and $\epsc^{2} \le \eps \le \frac{1}{8f}$. 
See Figure \ref{fig:gray_cell_partition}.
 In total, there are $2/\epsc$ gridlines.
Thus, $\area(\mathcal{P}_{cut}) \le 8 f\epss/\epsc$ and all these items can be packed separately into a total of $2/\epsc$ cells.
\end{proof}

We will repack the polygons from $\mathcal{P}_{cut}$ later. Now we again partition the cells into gray, white, and black as before.

\begin{defn} Let $G_{\ell,\ell'}\in\G$ for some $\ell,\ell'\in\{0,1,\ldots,1/\epsc-1\}$.
	The cell $G_{\ell,\ell'}$ is
	\begin{itemize}
		\item \emph{white} if $G_{\ell,\ell'}$ does not intersect with any polygon
		$P_{i}\in\C_{L}^{*}$,
		\item \emph{black} if $G_{\ell,\ell'}$ is contained in some polygon $P_{i}\in\P_{L}^{*}$,
		\item \emph{gray }if $G_{\ell,\ell'}$ is neither white nor black.
	\end{itemize}
\end{defn} 
Let $\G_w, \G_b,\G_g$ be the set of white, black, and gray cells, respectively. Like before, the area of the gray grid cells is intuitively lost. However, we show that there are only few gray grid cells and that their area is small.

\begin{lem} \label{lem:gray-circles-polygon}
There are at most $\frac{80 f}{\epsl \cdot \epsc}$ gray grid cells and their total area is at most~$80 f \epsc/ \epsl$.
 We can compute $\G_g$ in polynomial time.
\end{lem} 
\begin{proof}
	The perimeter of a single large convex polygon $P_i$ is upper bounded by $2 \pi r^{out} \le 2 f \pi r^{in}$. Let us divide the perimeter of $P_i$ into at most $\lceil\frac{2 f \pi r^{in}}{\epsc}\rceil$ pieces of length $\epsc$.
	One such piece intersects at most $9$ grid cells, so the total area of grey cells created by $P_i$ is upper bounded by $9 \epsc^2 \lceil\frac{2 f \pi r^{in}}{\epsc}\rceil$, which is roughly  $\le 20f\pi r^{in}\epsc$. 
	
	From the proof of \Cref{lem:gray-circles} we know that all cells inside the circle with a radius of $r^{in}- 3\epsc$ are black. Hence we get that the fraction \[
	\frac{\area(\G_{g, P_L})}{\area(\G_{b, P_L})} \leq \frac{20 f \epsl \epsc}{(\epsl - 3 \epsc)^2} \leq \frac{80 f \epsc}{\epsl}.
	\]
	The second inequality follows from the fact that $\epsc \leq \frac{1}{6} \epsl$ for any $\eps \leq 1/2$.  
	The total area of all black cells is at most one, so we get that the total area of all gray cells is at most $\frac{80 f \epsc}{\epsl}$.
 And the total number of gray cells is at most $\frac{80 f}{\epsl \cdot \epsc}$.

\end{proof}

Similarly, as before, we show that the number of white cells is much larger than the number of gray cells.

\begin{lem} \label{lem:white-polygons}
	The total area of white grid cells $\area(\G_{w})$ is at least 
 $\left ( \frac{\pi^2 \sin^2(\alpha)}{q^2 t^2} \right )\cdot \epsl^2$. 
 We can compute $\G_w$ in polynomial time.
\end{lem}
\begin{proof}
First, note that since the polygon is $t$-regular and has at most $q$ sides, 
	$2\pi r^{in} \le q s_{\max} \le q t s_{\min}$, where $s_{\max}, s_{\min}$ denote the maximum and minimum length of the sides of a large polygon.
	Now since $r^{in} \ge \epsl$ for large polygons, we have $s_{\min} \ge \frac{2\pi \epsl}{q t}=:\ell^*$.

 Consider a corner (w.l.o.g., left bottom corner) $v_{lb}$ of the knapsack. Let  $S_w$ be an axis-aligned square region of side length $\ell^* \sin (\alpha) /2$ such that $S_w$ is contained in the knapsack and shares $v_{lb}$. $S_w$ lies inside the bin for small enough values of $\epsl$ and in turn $\eps$. We can show that no large polygons intersect $S_w$ and thus all grid cells in $S_w$ are white.

For contradiction, assume some large polygon $P$ intersects $S_w$. 
 There are two cases. 

\begin{figure}
	\begin{center}
		\includegraphics[width=12cm]{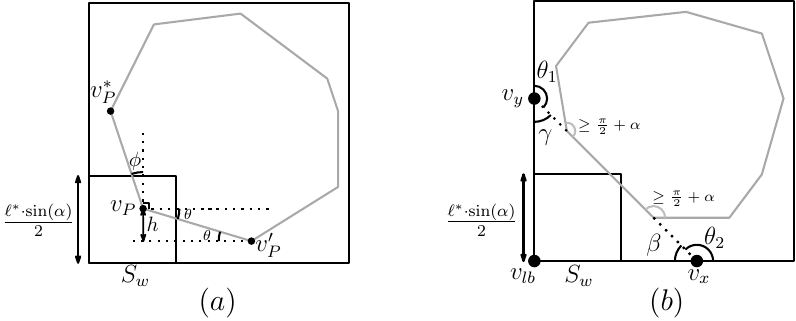}
		\caption{Cases when a polygon $P$ intersects $S_W$: (a) $S_W$ contains a vertex $v_P$ of $P$, (b) $S_W$ contains no vertices of $P$.}
		\label{fig:polygon_white_cell}
	\end{center}
\end{figure}

 In the first case, $P$ has a vertex $v_P=(x_1,y_1)$ inside $S_w$, so $y_1 <\ell^* \sin(\alpha)/2$. See Figure~\ref{fig:polygon_white_cell}(a) for an illustration. Let $v'_P, v^*_P$ be the adjacent vertices of $v_P$. Let $\theta$ (resp. $\phi$) be the acute angle that $v_P v'_P$ (resp. $v_Pv^*_P$) makes with $x$-axis (resp. $y$-axis). As $\angle v^*_P v_P v'_P \ge \pi/2+\alpha$, either $\theta \ge \alpha/2$ or $\phi \ge \alpha/2$. W.l.o.g, assume $\theta \ge \alpha/2$ (it is analogous for   $\phi \ge \alpha/2$).
Let $v'_P=(x_2, y_2)$ and $h=y_1-y_2$. Then we claim that $y_1 \ge h \ge \ell^* \sin(\theta) \ge \ell^* \sin(\alpha/2) \ge \ell^* \sin(\alpha)/2$, leading to a contradiction. 

In the other case, if some polygon edge $e$ cuts through $S_w$ then we can extend edge $e$ till the boundary and say it touches $x$-axis at $v_x$ and $Y$-axis at $v_y$. See Figure \ref{fig:polygon_white_cell}(b) for an illustration.  
Let $\beta:=\angle v_y v_x v_{lb}$  and $\gamma:=\angle v_x v_y v_{lb}$.
Let  $\theta_1, \theta_2$ be the obtuse angles made by $e$ with $y$-axis and $x$-axis resp., i.e., $\theta_1=\pi-\gamma$ and $\theta_2=\pi-\beta$. Then $\theta_1\ge \pi/2+\alpha$ and $\theta_2\ge \pi/2+\alpha$. 
Hence, $\beta = \theta_1-\pi/2 \ge \alpha$. Similarly, $\gamma \ge \alpha$. Thus, $\min\{\overline{v_{lb}v_x},\overline{v_{lb}v_y}\}\ge \ell^* \sin (\alpha)$. 
However, then $S_w$ cannot intersect $e$, leading to a contradiction. 
	
	We can now argue this analogously for all the corners of the knapsack and the lemma is proved. 
\end{proof}

Since there are much more white grid cells than gray grid cells, we can show that we can pack almost all (small) polygons from $\P^*_S$ into the white grid cells. Note that  $\P^*_S$ includes the polygons from $\P_{cut}$ that we removed earlier.

\begin{lem} \label{lem:small-circles-white-grid-cells}There
	is a set $\P'_S \subseteq \P^*_S$ of small polygons with
	$\profit(\P'_S)\ge (1-O(\eps)) \profit(\P^*_S)$
	such that there is a packing for $\P'_S$ using the grid cells in $\G_w$ only.
\end{lem}
\begin{proof}
	We delete all polygons in the 
$\eps$ fraction of white cells with the smallest total
	weight among all white cells, incurring an $\eps$ fraction loss in the weight. 
 This makes $\left ( \frac{\eps \pi^2 \sin^2 \alpha}{q^2 t^2} \right )\cdot (\epsl/\epsc)^2$
 number of cells available for packing polygons in the gray cells and $\P_{cut}$. 
	Now, the total number of cells to pack polygons in the gray cells and $\P_{cut}$ is at most 
$2/\epsc+ \frac{80f}{\epsl \epsc}$.
 So, we require the condition $2/\epsc+ \frac{80f}{\epsl \epsc} \le \left ( \frac{\eps \pi^2 \sin^2 \alpha}{q^2 t^2} \right )\cdot (\epsl/\epsc)^2$, or equivalently we need $2 \epsl^3+\epsl^2 \cdot 80f \le \eps (\frac{\pi^2 \sin^2 \alpha}{q^2 t^2})$. 
 For this to hold, we fix $\eps$ such that $\eps \le \frac{\pi^2 \sin^2(\alpha)}{q^2t^2(2+80f)}$.
 Then, $2 \epsl^3+\epsl^2 \cdot 80f \le 2 \epsl^2+\epsl^2 \cdot 80f \le \eps^2(2+80f) \le \eps \cdot \eps(2+80f) \le \eps (\frac{\pi^2 \sin^2 \alpha}{q^2 t^2})$.
\end{proof}

Finally, we sacrifice a factor of $1-O(\eps)$ in the profit from the small polygons in $\P'_S$ such that we can ensure that
the remaining small polygons fit even into smaller (white) grid cells of size $(1-\eps)\epsc \times(1-\eps)\epsc$ each.

\begin{lem} \label{lem:grid-resource-augmentation}There is a set
	of small polygons $\P''_{S}\subseteq\P'_{S}$ such that $w(\P''_{S})\ge(1-O(\eps))w(\P'_{S})$
	and it is possible to place the polygons in $\P''_{\sml}$ non-overlappingly
	inside $|\G_{w}|$ square knapsacks of size $(1-\eps)\epsc \times(1-\eps)\epsc$
	each. \end{lem}
\begin{proof}
	Consider a square knapsack in $\G_{w}$ of side length $\epsc$.
        Similar to Lemma \ref{lem:grid-resource-augmentation-1}, the probability that a small polygon is intersected by a random strip of width $\eps \epsc$ is at most $(\eps\epsc+2f\epss)/\epsc \le 3\eps$.
	Remove all polygons intersected by a random strip of width $\eps$ for both directions, we incur at most a $(1-O(\eps))$-factor loss in the weight.
\end{proof}
This completes the proof of Lemma~\ref{lem:structured-packing-polygon} with $\P_{S}:= \P''_{S}$.

\section{Conclusion}
\label{sec:conc} 
We almost settle the approximability of the geometric knapsack problem in the setting of 
packing spheres into a hypercube knapsack. However, it remains an open problem whether rational coordinates always suffice in an optimal packing.
If not, it would be an interesting question to determine the 
best approximation ratio one can obtain if we allow only rational coordinates for the centers of the circles (while the optimal packing has no such restrictions).
It would be also interesting to obtain a \textsf{PTAS} for the case of $d$-dimensional fat convex objects.
Another interesting but difficult open question is whether the case of convex but not necessarily fat input objects in the plane admits a \textsf{PTAS}. The best \agtwo{known} result for this setting is only an $O(1)$-approximation in quasi-polynomial time (assuming polynomially bounded integral input data \cite{MerinoW20}).  Already for the special case of axis-parallel rectangles, it is open whether a \textsf{PTAS} exists. 

\bibliography{references}

\appendix


\section{Shifting step}\label{sec:shifting}
\begin{lem}
	\label{lem:shifting}
	Let a given packing for a knapsack packing problem with input set \(\C\) of items be $\mathcal{P}$ and \(r\) be a total order on these objects and \(w_i\) be the profit associated with the \(i\)-th object. Let $\mathcal{C}_{\mathcal{P}}$ be the set of the items selected by $\mathcal{P}$ and we have $w(\mathcal{P}) = \sum_{C_{i}\in \mathcal{C}_{\mathcal{P}}} w_{i}$. Define $\rhoi{k} = h(\eps, k)$ where $h: \mathbb{R} \rightarrow \mathbb{R}$ is a given decreasing function. For each integer \(k\geqslant 1\), define the following sets of items:
	\(
	\mathcal{C}_k = \{C_i \in \mathcal{C} \mid r_{i} \in (\rho_{k},\rho_{k-1}]\}.
	\)
	Then there exists a \(\tau \leqslant \frac{1}{\eps}\) such that \(w(\C_{\tau}\cap\C_{\mathcal{P}})\leqslant\eps w(\mathcal{P})\).
\end{lem}
\begin{proof}
We will partition the input set of $\mathcal{C}$ into large, medium, and small items using a partition function $\rho_k$ which we will design retrospectively after considering all the properties it needs to satisfy. The methodology is as follows, we partition the input $\mathcal{C}$ into several sets \(\mathcal{C}_k\)'s, and then select an appropriate index \(j \in [1/\eps]\) such that the sum of weights of all items from $\mathcal{P}$ in \( \mathcal{C}_j \leqslant \eps\cdot w(\mathcal{P})\) for a given constant $\eps>0$. We will declare this to be the medium items set $\mathcal{C}_M$ and discard it from $\mathcal{C}$.

Let $\eps>0$ be a positive real constant. We can assume that $\mathcal{P}$ has more than $\frac{1}{\eps}$ items. Otherwise, it is trivial.
Now, we define $\rho_0=1$, and $\rhoi{k} := h(\eps, k) \cdot \rho_{k-1}$ where $h: \mathbb{R} \rightarrow \mathbb{R}$ is the given decreasing function. For each integer $k \geqslant 1$ we define the following sets of items:
\begin{align*}
\mathcal{C}_k = \{C_i \in \mathcal{C} \mid r_{i} \in (\rho_{k},\rho_{k-1}]\}.
\end{align*}

Let $\tau$ be the smallest index such that the total weight, $w(\mathcal{C}_{\tau})$, of the items in $S_{\tau}$ is at most $\eps w(\mathcal{P})$. Observe that $\tau \leqslant \frac{1}{\eps}$ due to the pigeonhole principle. Partition the set $\mathcal{C}$ of items into three groups: the $\textit{large}$ items $\mathcal{C}_L = \{C_{i} \in \mathcal{C} \mid r_{i}>\rho_{\tau-1}\}$, the $\textit{medium}$ items $\mathcal{C}_M = \{C_{i} \in \mathcal{C} \mid \rho_{\tau} < r_{i} \leqslant \rho_{\tau-1}\}$ and the $\textit{small}$ items $\mathcal{C}_S = \{C_{i} \in \mathcal{C} \mid r_{i} \leqslant \rho_{\tau}\}$. Similarly, define sets $\mathcal{C}_L^* = \mathcal{C}_L \cap \mathcal{P}$, $\mathcal{C}_M^* = \mathcal{C}_M \cap \mathcal{P}$, and $\mathcal{C}_S^* = \mathcal{C}_S \cap \mathcal{P}$. Note that $\mathcal{C}_M^* = \mathcal{C}_{\tau}$, so the total weight of the medium items is at most $\eps\cdot w(\mathcal{P})$.
\end{proof}

\section{\textsf{PTAS} for $d$-dimensional hyperspheres}
In this section we prove \Cref{thm:d-dim-ptas}.
	In order to extend our techniques to \(d\)-dimension, we start with first creating partitions over the \(d\)-dimensional hyperspheres, similar to Lemma \ref{lem:create-gap},  by setting \(\eps^{j}=(\eps^{j-1})^{12d}\) and assume that medium hyperspheres contribute negligible profit to $\OPT$.
	
	Further, we notice that each large hypersphere has a volume of at least \( \dfrac{\pi^{\frac{d}{2}}}{\Gamma\left(\frac{d}{2} + 1\right)} \cdot \epsl^d
	\). Thus, there can be at most  \(\left(\frac{\pi^{\frac{d}{2}}}{\Gamma\left(\frac{d}{2} + 1\right)} \cdot \epsl^d\right)^{-1}\)  
 large hyperspheres in any feasible solution. Using the techniques as in \Cref{subsec:Large_Guess}, we can guess and verify the coordinates of the large hyperspheres in \(\mathcal{C}_{L}^*\) in polynomial time. Also observe that the guessed values $\tilde{x}_{i}^{(1)},\tilde{x}_{i}^{(2)},\dots,\tilde{x}_{i}^{(d)}$
	yield an estimate for $\hat{x}_{i}^{(1)},\hat{x}_{i}^{(2)},\dots,\hat{x}_{i}^{(d)}$ up to a (polynomially small) error of $\eps/n$.
	
	In order to find a feasible packing of the small hyperspheres we will replicate techniques from \Cref{subsec:placing_small}.  Define \(\epsc:=\epsl^{6d}\)(i.e, \(\epss=\epsl^{12d}=\epsc^2\)). To subdivide the knapsack into \(1/\epsc^d\) $d$-dimensional grid cells of length \(\epsc\). We denote a grid cell by \(G_{\ell^{(1)},\dots,\ell^{(d)}}:=[\ell^{(1)}\cdot\epsc,(\ell^{(1)}+1)\cdot\epsc)\times\dots\times[\ell^{(d)}\cdot\epsc,(\ell^{(d)}+1)\cdot\epsc)\). We define the set of grid cells by \(\G:=\{G_{\ell^{(1)},\dots,\ell^{(d)}}:\ell^{(1)},\dots,\ell^{(d)}\in\{0,1,\dots,1/\epsc-1\}\}\). To proceed further, we first prove structural properties for the packing of the hyperspheres in the grids. For any set \(\mathcal{S}\) of circles or grid cells, we define \(\area(\mathcal{S})\) to be the total volume of the elements in \(\mathcal{S}\).
	\begin{lem}\label{lem:small-hypersphere-grid-cells}
		Let \(\mathcal{C}_{cut}\subseteq\C_{S}^{*}\) be the set of all small hyperspheres in \(\C_{S}^{*}\) that intersect more than one grid cell. We have that \(\area(\mathcal{C}_{cut})\leqslant4d\epss/\epsc\leqslant d\eps\epsl^d/2^{5d-3}\)
	\end{lem}
	\begin{proof}
		For each of the \(i\)-th dimension for \(1\leqslant i \leqslant d\), and for each of the gridline \(\ell\) in this dimension, each \(C\in\C_{cut}\) must lie in a strip with the \(i\)-th dimension being \(4\epss\) and all other dimensions being 1. Totally, there are \(d/\epsc\) gridlines. Thus, \(\area(\C_{cut})\leqslant4d\epss/\epsc=4d\epsl^{6d}\leqslant d\eps\epsl^d/2^{5d-3}\). The last inequality follows from the fact \(\epsl\leqslant \eps\leqslant 1/2\).
	\end{proof}
	We partition \(\G\) into \emph{white}, \emph{black}, and \emph{gray} cells with respect to the position of the large hyperspheres similarly as the case of circles. We give bounds similar to \Cref{lem:gray-circles} and \Cref{lem:white-circles} in the case of $d$-dimensions for \emph{gray} and \emph{white} cells.
	\begin{lem}\label{lem:gray-circles-d}
		The total volume of gray cells \(\area(\G_{g})\) is at most \(\dfrac{(d+1)(3^d-1)\eps\epsl^d}{2^{5d-3}}\). We can compute \(\G_g\) in polynomial time.
	\end{lem}
	\begin{proof}
		Consider a large hypersphere \(C_L\) with radius \(r_L\geqslant\epsl\). Let \(\G_{g, C_L},\G_{b, C_L}\) be the set of gray and black grid cells that overlap with \(C_L\). Note that, for \(C_L\) we only have an estimate \(\tilde{x}_{C}^{(1)},\dots,\tilde{x}_{C}^{(d)}\) of the coordinate of its center up to a \(\emph{polynomially}\) small error of \(\eps/n\).
		
		First we claim that all the cells intersected by a hypersphere centered at  \((\tilde{x}_{C}^{(1)},\dots,\tilde{x}_{C}^{(d)})\) of radius $r_L- (d+1)\epsc$ are in $\G_{b, C_L}$. 
		For contradiction, assume there is a gray cell $S_g$ intersected by a circle $C$ centered at  \((\tilde{x}_{C}^{(1)},\dots,\tilde{x}_{C}^{(d)})\) of radius $r_L- (d+1)\epsc$.
		Then, the distance of any point in $S_g$ from the center of the circle is at most $r_L- (d+1)\epsc +\sqrt{d}\epsc < r_L-\eps/n$. Thus, $S_g$  is completely contained within $C$.
		
		Similarly, we can show that the cells that are outside of a circle centered at  \((\tilde{x}_{C}^{(1)},\dots,\tilde{x}_{C}^{(d)})\) of radius $r_L+ (d+1)\epsc$ are not in $\G_{b, C_L}$.
		
		Hence, 
		\begin{align*}
		\area(\G_{b, C_L}) &\geqslant \frac{\pi^{\frac{d}{2}}}{\Gamma\left(\frac{d}{2} + 1\right)} \cdot (r_L-(d+1)\epsc)^d \\
		\area(\G_{g, C_L}) &\leqslant \frac{\pi^{\frac{d}{2}}}{\Gamma\left(\frac{d}{2} + 1\right)} \cdot (r_L+(d+1)\epsc)^d - \frac{\pi^{\frac{d}{2}}}{\Gamma\left(\frac{d}{2} + 1\right)} \cdot (r_L-(d+1)\epsc)^d \\ 
		&\leqslant \frac{\pi^{\frac{d}{2}}}{\Gamma\left(\frac{d}{2} + 1\right)} \cdot 2(d+1)((3/2)^d-(1/2)^d)\epsc r_{L}^{d-1}
		\end{align*}
		The inequality follows from the fact, \(r_{L}/2\geqslant(d+1)\epsc\)\\
		
		Hence we can bound the ratio of \(\frac{\area(\G_{g, C_L})}{\area(\G_{b, C_L})}\) as
		\[
		\begin{aligned}
		\frac{\area(\G_{g, C_L})}{\area(\G_{b, C_L})} &\leqslant \frac{\frac{\pi^{\frac{d}{2}}}{\Gamma\left(\frac{d}{2} + 1\right)} \cdot 2(d+1)((3/2)^d-(1/2)^d)\epsc r_{L}^{d-1}}{\frac{\pi^{\frac{d}{2}}}{\Gamma\left(\frac{d}{2} + 1\right)} \cdot (r_L-(d+1)\epsc)^d} \\
		&\leqslant \frac{2(d+1)((3/2)^d-(1/2)^d)\epsc r_{L}^{d-1}}{(r_L/2)^2} \\
		&\leqslant \frac{2(d+1)(3^d-1)\epsc}{r_L} \leqslant \frac{2(d+1)(3^d-1)\epsc}{\epsl} \\
		&= 2(d+1)(3^d-1)\epsl^{6d-1} \leqslant \frac{(d+1)(3^d-1)\eps\epsl^d}{2^{5d-3}}
		\end{aligned}
		\]
		
		Inequalities follow from the fact that $\epsl \leqslant \eps \leqslant 1/2$ and $r_L \geqslant \epsl$.  The last inequality uses the fact that $\eps \leqslant 1/2$.
		Therefore, $\sum_{C \in \C_{L}^{*}} {\area(\G_{g, C_L})} \leqslant \bigg(\dfrac{(d+1)(3^d-1)\eps\epsl^d}{2^{5d-3}}\bigg) \cdot \sum_{C \in \C_{L}^{*}}  {\area(\G_{b, C_L})} \leqslant \dfrac{(d+1)(3^d-1)\eps\epsl^d}{2^{5d-3}}$.
	\end{proof}
	\begin{lem}\label{lem:white-circles-d}
		The total volume of the white grid cells \(\area(\G_w)\) is at least \(\frac{\epsl^d}{2^d}\), We can compute \(\G_w\) in polynomial time.
	\end{lem}
	\begin{proof}
		Similarly arguing as in \Cref{lem:white-circles}, we can show that no hypersphere can intersect a $d$-dimensional bin of side length \(\epsl/4\) at each corner. Hence the total volume of white grid cells must be at least \(2^d\cdot(\epsl/4)^d=(\epsl/2)^d\)
	\end{proof}
	
	We note here that the total volume of gray cell and \(\C_{cut}\) is at most \(d\eps\epsl^d/2^{5d-3}+\dfrac{(d+1)(3^d-1)\eps\epsl^d}{2^{5d-3}} \leqslant \eps (\epsl/2)^d\). Hence similar to \Cref{lem:small-circles-white-grid-cells} we state the following lemma for $d$-dimensional case also using \Cref{lem:small-hypersphere-grid-cells}, \Cref{lem:gray-circles-d}, and \Cref{lem:white-circles-d}.
	\begin{lem}
		\label{lem:small-circles-white-grid-cells-d}
		There
		is a set $\C'_S \subseteq \C^*_S$ of small hyperspheres with
		$p(\C'_S)\geqslant (1-\eps) p(\C^*_S)$
		such that there is a packing for $\C'_S$ using the grid cells in $\G_w$ only.
	\end{lem}
	We extend \Cref{lem:grid-resource-augmentation} for the $d$-dimensional case.
	\begin{lem} \label{lem:grid-resource-augmentation-d}
		There is a set of small hyperspheres $\C''_{S}\subseteq\C'_{S}$ such that $p(\C''_{S})\geqslant(1-\eps)p(\C'_{S})$
		and it is possible to place the hyperspheres in $\C''_{S}$ non-overlappingly
		inside $|\G_{w}|$ $d$-dimensional knapsacks of size $(1-\eps)\epsc \times\cdots \times(1-\eps)\epsc$
		each. \end{lem}
	\begin{proof}
			Consider a $d$-dimensional knapsack in $\G_{w}$ of side length $\epsc$.
		Similar to Lemma \ref{lem:grid-resource-augmentation-1}, by removing all circles intersected by a random strip of width $\eps \cdot \epsc$ in all the $d$ directions, we incur a loss of at most $d(d r/\epsc + \eps)$ fraction of optimal profit. As $r\le \epss$, we incur a loss of at most $O(\eps)$ fraction of profit as $\epss = \epsc^2 = \epsl^{12d}$.
	\end{proof}
	We further extend \Cref{lem:structured-packing} to $d$-dimensional case using \Cref{lem:grid-resource-augmentation-d}
	\begin{lemma}\label{lem:structured-packing-d}
		In polynomial time, we can compute a set of grid cells $\G_w$ such that no grid cell in $\G_w$ intersects with any
		hypersphere $C_i \in\C_{L}^{*}$ for any legal placement of $C_i$. Moreover, there is a set of small circles $\C_S \subseteq \C_{S}^{*}$ such that $w(\C_{S}) \geqslant (1-\eps)w(\C_{S}^{*})$ and
		the circles in $\C_S$ can be packed non-overlappingly inside $|\G_w|$ grid cells of size $(1-\eps) \epsc \times\dots\times(1-\eps) \epsc$ each.
	\end{lemma}
		Using it, we compute an approximation to the set $\C_S$ via an extension of our result from \Cref{thm:fat-ra} to $d$-dimensions and setting fatness $f$ to 1 for spheres.

	We pack the computed hyperspheres into our grid cells $\G_w$, denote them by $\C'_S$. In particular, they do not intersect any of the
	large hyperspheres in $\C_{L}^{*}$ in any legal placement of them. Our solution (corresponding to the considered guesses) consists of $\C_{L}^{*}\cup \C'_S$.

In order to further improve the precision of the large hyperspheres, we replicate techniques from \Cref{subsec:Improve_Large} which completes the proof of \Cref{thm:d-dim-ptas}.

\end{document}